\pdfoutput=1
\documentclass[a4paper,UKenglish,cleveref, autoref, thm-restate,authorcolumns, colorlinks]{lipics-v2021}
\usepackage{hyperref}
\usepackage{amsfonts}
\usepackage{float}
\usepackage{amsthm}
\usepackage{amsmath}
\usepackage{graphicx}
\usepackage{cleveref}
\usepackage{tikz}
\usepackage{tikz-cd} 
\usepackage{apxproof}
\usepackage{comment}
\usepackage[small,caps]{complexity}
\usepackage{pgf}
\usepackage{xcolor}
\usetikzlibrary{decorations.pathreplacing,calligraphy}
\usetikzlibrary{calc,patterns}
\usetikzlibrary{arrows,automata,positioning}

\newtheoremrep{theorem}{Theorem}[section]
\newtheoremrep{corollary}[theorem]{Corollary}
\newtheoremrep{proposition}[theorem]{Proposition}
\newtheoremrep{lemma}[theorem]{Lemma}

\newcommand\dia[2]{\ensuremath{dia_#2(#1)}}
\newcommand\circn[1]{\ensuremath{({#1})}}
\newcommand\indx[2]{\ensuremath{\textrm{Index}(#1,#2)}}
\renewcommand\implies{\ensuremath{\Rightarrow}}
\newcommand\as{\ensuremath{\alpha}}
\newcommand\bs{\ensuremath{\beta}}
\newcommand\calA{\mathcal{A}}
\newcommand\calB{\mathcal{B}}
\newcommand\calT{\mathcal{T}}
\newcommand\calS{\mathcal{S}}
\newcommand\calF{\mathcal{F}}
\newcommand{\tdots}{\ensuremath{\,.\,.\,}}
\newcommand\N{\mathbb{N}}

\newcommand\lborder{\textit{lborder}}
\newcommand\rborder{\textit{rborder}}
\newcommand\interior{\textit{interior}}
\newcommand\dlen{d_{\mathit{len}}}
\newcommand\dlcs{d_{\mathit{lcs}}}
\newcommand\ddl{d_{\mathit{dl}}}
\newcommand\omicron{\mathit{o}}
\newcommand\dom{\mathit{dom}}

\renewcommand\){\right)}

\newtheorem{definition}[theorem]{Definition}
\newtheorem{claim}{Claim}
\newtheoremrep{claim}[theorem]{Claim}

\title{Edit Distance of Finite State Transducers}

\author{C. Aiswarya}{Chennai Mathematical Institute, India and IRL ReLaX, CNRS France}{aiswarya@cmi.ac.in}{https://orcid.org/0000-0002-4878-7581}{}
\author{Amaldev Manuel}{Indian Institute of Technology Goa}{amal@iitgoa.ac.in}{}{Supported by the DST SERB MATRICS grant MTR/2022/000628 \textit{Deciding closeness of finite state transducers}.} 
\author{Saina Sunny}{Indian Institute of Technology Goa}{saina19231102@iitgoa.ac.in}{https://orcid.org/0009-0005-1366-0168}{}

\authorrunning{C. Aiswarya, A. Manuel, and S. Sunny} 

\Copyright{C. Aiswarya, A. Manuel, and S. Sunny} 



\category{} 

\nolinenumbers 

\hideLIPIcs

\begin{document}
\maketitle

\begin{abstract}
 We lift  metrics over words to  metrics over  word-to-word  transductions, by defining the distance between two transductions as the supremum of the distances of their respective outputs over all inputs. This allows to compare transducers beyond equivalence. 
 
 Two transducers are \emph{close} (\textit{resp.}~$k$-close) with respect to a metric if their distance is finite (\textit{resp.}~at most $k$). Over integer-valued metrics computing the distance between transducers is equivalent to deciding the closeness and $k$-closeness problems. For common integer-valued edit distances such as, Hamming, transposition, conjugacy and Levenshtein family of distances, we show that the closeness and the $k$-closeness problems are decidable for functional transducers. Hence, the distance with respect to these metrics is also computable. 

Finally, we relate the notion of distance between functions to the notions of diameter of a relation and index of a relation in another. We show that computing edit distance between functional  transducers is equivalent to computing diameter of a rational relation and both are a specific instance of the index problem of rational relations. 
\end{abstract}

\section{Introduction}
For meaningfully comparing two words (or sequences, vectors, functions, etc.), it is often necessary to have a measure that quantifies their (dis)similarity. It usually consists of associating a nonnegative integer to two words that indicates how different they are from each other. This usually defines
a distance between words, the most popular of which are edit distances. It is the minimum number of edit operations required to transform one word into another. These operations typically include inserting or deleting a letter, substituting a letter with another,  swapping adjacent letters (transpositions), and cyclic shifts. 
Edit distances are studied in coding   \cite{levenshtein1966binary,ullman1966near}, parsing   \cite{aho1972minimum,pighizzini1992parallel}, speech recognition \cite{okuda1976method,ackroyd1980isolated},  molecular biology \cite{eppstein1990efficient, karp1993mapping} etc. Interesting combinatorial problems on words such as the computation of longest common subsequences can be reduced to computing edit distances \cite{apostolico1987longest}. 
 For a detailed overview of the history and applications of edit distances,  see  \cite{kruskal1983overview}.

The notion of distance between two words can be lifted naturally to distance between a word and a set of words, or between two sets of words, and so on. There is a long line of research of  this kind: computing the edit distance between two languages --- usually defined as the smallest distance between any two pairs from the respective sets.  
 It could be between a word and a regular language \cite{OrderNCorrectionRegular} \cite{StringFiniteAutomaton}, two regular languages \cite{EditDistanceOfRegularLanguages}, a regular language and itself \cite{EditDistanceRegularLanguage}, or a regular language and a context-free language \cite{EditDistanceRegularCFL}. In all these settings there are efficient algorithms for computing the edit distances. 

In this paper, we study the distance between two word-to-word functions (transductions) given by finite state transducers, i.e., automata with output. 
A \emph{transducer} is a machine that reads an input word and produces one or more output words. They were one of the earliest machines to be studied in automata theory, and in fact finite state automata were first defined as a special case of finite state transducers \cite{RabinScott}.
Finite state transducers are used in a variety of software and hardware systems such as encoders, decoders, demuxers, spell checkers, text normalizers, schema translators, template code generators, etc. 
 
\begin{figure}[htbp]
\centering
\scalebox{.8}{
\begin{subfigure}[b]{0.5\textwidth}
\begin{tikzpicture}[shorten >=1pt,node distance=3 cm,on grid,auto] 
   \node[state,initial,accepting,initial text=] (q_0)   {$q_0$}; 
   \node[state,accepting] (q_1) [right=of q_0] {$q_1$};
    \path[->] 
    (q_0) edge [bend left] node {$a|a, b|b$} (q_1)
    (q_1) edge [bend left] node  {$a|\epsilon, b|\epsilon$} (q_0);
\end{tikzpicture}
\caption{$\calT_1$} 
\end{subfigure}
\hfill
\begin{subfigure}[b]{0.5\textwidth}
\begin{tikzpicture}[shorten >=1pt,node distance=3cm,on grid,auto] 
   \node[state,initial,accepting,initial text=] (q_0)   {$q_0$}; 
   \node[state,accepting] (q_1) [right=of q_0] {$q_1$};
    \path[->] 
    (q_0) edge [bend left] node {$a|\epsilon, b|\epsilon$} (q_1)
    (q_1) edge [bend left] node  {$a|a, b|b$} (q_0);
\end{tikzpicture}
\caption{$\calT_2$}
\end{subfigure}
 \hfill 
\begin{subfigure}[b]{0.5\textwidth}
\begin{tikzpicture}[shorten >=1pt,node distance=3 cm,on grid,auto] 
   \node[state,initial,accepting,initial text=] (q_0)   {$q_0$}; 
    \path[->] 
    (q_0) edge [loop above] node {$a|a$} (q_0)
            edge [loop below] node {$b|\epsilon$} (q_0);
\end{tikzpicture}
\caption{$\calT_3$}
\end{subfigure}
}
\caption{$\calT_1$ outputs letters at the odd positions, $\calT_2$ outputs letters at the even positions and $\calT_3$ outputs only $a$'s.}
\label{UnboundedFig}
\end{figure}
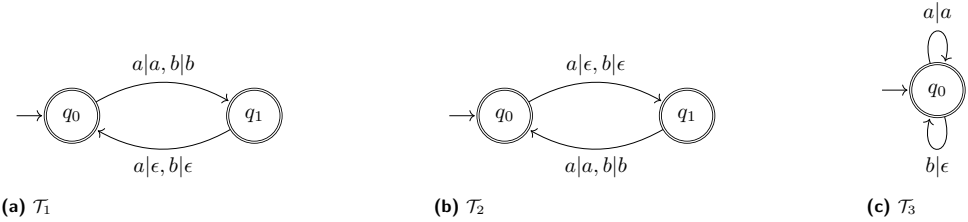 
Our aim is to develop a framework to meaningfully compare two transductions beyond equivalence. 
Consider the functions given by the transducers in \Cref{UnboundedFig}. The transducers $\calT_1$ and  $\calT_2$ output the letters at the odd and even positions respectively, while the transducer $\calT_3$ erases $b$'s in the input. If we were to find the odd one among these three functions, arguably $\calT_3$ will be picked, with the length of the respective output on any input deviating significantly from that of the others. Our aim is to define a measure that quantifies such distances. 

If we have a metric to compare the output words, we can extend it to transductions as follows. The distance between two transductions is the least upper bound of the distances between their respective outputs on any input word. We assume that their domains are the same, and we set the distance to infinity if this is not the case.  We say that two transductions are close if their distance is finite, and they are $k$-close if their distance is at most $k$. We may simply say that two transducers are close (or $k$-close) instead, to mean that the transductions defined by these transducers are close (or $k$-close).

We are interested in the following question: \textit{Given two finite state transducers, are the transductions defined by them close (or simply are the transducers close)?}
Clearly, deciding closeness is a boundedness problem. We show closeness as well as  $k$-closeness  are decidable for various edit distance metrics, in particular Hamming (letter-to-letter substitutions), transposition (swapping adjacent letters), conjugacy (only cyclic shifts) and Levenshtein family of distances --- Longest common subsequence (insertion and deletion), Levenshtein (insertion, deletion and substitution), and Damerau-Levenshtein (insertion, deletion, substitution and adjacent transposition). It turns out that computing distance between transducers is equivalent to deciding closeness and $k$-closeness over integer-valued metrics (see \Cref{prop:computedistance}).  Hence for the edit distances mentioned above, the distance between transducers is computable.

A related notion is that of \textit{diameter} of a relation. We define it to be the supremum of the distance of every pair in the relation. We are interested in computing the diameter of rational relations over words, that is those given by (not necessarily functional) finite state transducers. A rational relation is said to have \emph{bounded diameter} (\textit{resp.}~$k$-bounded diameter) if the diameter of the relation is finite (\textit{resp.}~at most $k$). It turns out that for every pair of transductions $\calT_1$ and $\calT_2$ there is a rational relation $R$ such that for every metric, the diameter of $R$ is same as the distance between $\calT_1$ and $\calT_2$. In fact, the converse is also true by virtue of Nivat's theorem (see \Cref{nivat}).

Another related notion is that of the  \emph{index of a rational relation in the composition closure of another}. Let $R,S$ be a rational relation over words. The index of $R$ in the composition  closure of $S$ is defined to be the smallest integer $k$ such that the relation $R$ is contained in the $k$-fold composition of $S$. If such a $k$ exists we say that $R$ has the \textit{finite index property} in the composition closure of $S$.  We show that the finite index property is undecidable for arbitrary rational relations. However, if $S$ is a metrizable relation (see \Cref{metrizable}) w.r.t.~the edit distances mentioned above, the index of $R$ in the composition closure of $S$ is computable. 

\medskip

Our decision procedure for $k$-closeness involves designing a weighted automaton that counts the number of edit operations for transforming one output to the other. We need to check whether there are input instances for which the weight is more than $k$. We extract a finite state automaton of size exponential in $k$ that achieves this (see \Cref{prop:kapprox}). This is a generic approach independent of the particular edit operations. However for Hamming and transposition distances, we have a direct polynomial time procedure for deciding $k$-closeness (see \Cref{kclosehammingtrans}).

Recall that deciding the closeness of transductions is the same as deciding whether the diameter of a rational relation $R$ is bounded. For the latter, consider a transducer recognising $R$. It turns out that if there are loops in this transducer that produce nonconjugate words (that are not cyclic shifts of each other) then such loops can be iterated to get unbounded diameter/distance. Thus a crucial ingredient in our decision procedure is checking for conjugacy of loops, which is decidable \cite{decidingconjugacy}. 
 For boundedness w.r.t.~Levenshtein distances, we show that this is also a sufficient condition (see \Cref{closewrtlevenshtein}). For conjugacy distance, we show that the diameter of a rational relation $R$ is bounded if and only if every pair in $R$ is conjugate (see \Cref{closewrtconjugacy}). Notice that this is not the case for arbitrary relations.  In the case of Hamming distance, which only includes substitutions, we show that it is sufficient to check if the pairs of words generated by the loops after some shifted delay are identical (see \Cref{Claim:hamming}). This also holds true for transposition distance, but additionally, we also need to check if the words are permutations of each other (see \Cref{Claim:trans}).

\subsection{Related Work}
\label{Subsec:relatedwork}
The \emph{adjacent functions} in \cite{ReutenauerS1991} is an analogous definition for closeness between transductions  with respect to prefix distance. Two functions $f,g : A^* \rightarrow B^*$ are \emph{adjacent} if $\sup \left \{\, d_p(f(w),g(w)) \mid w \in \dom(f) \cap \dom(g) \, \right \} < \infty$. Here, $d_p(u,v) = |u| + |v| - 2\max \{|z| \mid u,v \in zA^*\}$ denotes the \emph{prefix distance} between two words $u$ and $v$. The adjacency of two rational functions is used in deciding the sequentiality of a function. It is decidable to check if two given rational functions are adjacent or not \cite{ReutenauerS1991}(Proposition 1).

Another problem that is similar in spirit is the \emph{robustness problem}. We say a transducer $\calT$ is \emph{robust} w.r.t.~a distance $d$ if there is a nontrivial relation $R$ between the distance between two input words (say $d(u,v)$) and distance between their corresponding outputs on $\calT$ (say $d(\calT(u),\calT(v))$). For instance, $R$ could be \emph{Lipschitz continuity} --- there is some $k>0$ such that $d(\calT(u),\calT(v)) \leq k \cdot d(u,v)$, or \emph{locally Lipschitz continuity} --- there exists $b,k>0$ such that if $d(u,v)<b$ then $d(\calT(u),\calT(v)) \leq k \cdot d(u,v)$, etc. Sometimes, weaker notions of distance are considered (for instance by dropping the triangle inequality), and respective distances are called \emph{cost} or \emph{similarity} functions. The work \cite{RobustnessAnalysis} solves the locally Lipschitz continuity problem for sequential and unambiguous transducers using reversal bounded counter automata. The problem is shown to be undecidable for Lipschitz continuity even for deterministic transducers and the decidability is shown for the class that has a bound on the delay between input and output words \cite{LipschitzRobustness}.

 Frougny and Sakarovitch studied rational relations with bounded delay \cite{frougny1991rational}, which is actually our diameter problem for rational relations when the distance over words is measured by their length difference. A problem related to the diameter of a rational relation is \emph{almost reflexivity} of rational relations studied in \cite{choffrut2002distances}. A relation $R \subseteq A^* \times A^*$ is $k$-reflexive, for some integer $k \leq \infty$, if every element $u$ of the domain is at a distance at most $k$ from some element of the range $v$, with $(u,v) \in R$, and vice versa. The relation $R$ is almost reflexive if $k < \infty$. It is shown undecidable to check if a deterministic rational relation is almost reflexive, or $k$-reflexive, for any given integer k, with respect to the following -- Hamming,  prefix, suffix, subword and Levenshtein edit distances. It is shown decidable for synchronized rational relation w.r.t.~Hamming distance.

In 1966, Brzozowski raised the question of \emph{finite power property} on regular languages ---  it takes a regular language $L$ as input and asks whether there exists some positive integer $n$ such that $(L + \epsilon)^n = L^*$. It was solved in 1979 by Hashiguchi \cite{hashiguchi1979decision} and Simon \cite{simon1978limited}, independently. We study the \emph{finite index property} of a rational relation in the iterative composition of another relation.  Notice that the finite index property is different from the finite power property in two respects. One, it is over relations and not languages, and secondly and more importantly, the iteration is obtained by relation composition and not concatenation.

\subsection{Organisation of the Paper}
In Section~\ref{sec:prelim}, we recall the definitions of finite state transducers, metrics on words and edit distances. In Section~\ref{sec:3}, we define the notion of distance between transducers, the diameter of a rational relation, and the index of a rational relation in another. We also establish the relation between these notions and state our results in this section. In Section~\ref{Subsec:closeness}, we give the connections with conjugacy and the proof arguments remaining from Section~\ref{sec:3}. Finally, we conclude in Section~\ref{Sec:conclusion} with a short discussion on future directions. 

\section{Preliminaries}
\label{sec:prelim}
Let $A^*$ denote the set of all finite words over the alphabet $A$. We use $|w|$ to denote the length of the word $w$. 
Let $w[i \tdots j]$ denote the factor of $w$ from index $i$ to $j$ where $1 \leq i \leq j \leq |w|$. 
A transduction is a function from words to words. 

\subsection{Finite State Transducers}
\label{Subsec:fst}

The simplest form of a transducer is a deterministic finite state machine whose each transition and each final state is labelled by a possibly empty \emph{output word}. Formally, a {\em sequential transducer} $\calT = \langle \calA, \lambda, \omicron\rangle$ with input alphabet $A$ and output alphabet $B$ is a {\em deterministic} finite state automaton $\calA$ with two associated output functions $\lambda: \Delta \rightarrow B^*$ and $\omicron:F \rightarrow B^*$ where $\Delta$ and $F$ are the set of transitions and the set of accepting states of $\calA$ respectively. 

On an input word that is accepted by the automaton, we concatenate the output words produced by the transitions in the unique run of the machine and finally append the end-of-input word of the final state to obtain the output of the machine.
That is to say, if $\rho=\delta_1 \cdots \delta_n$ is the successful run of $\calA$ on a word $w \in A^*$, the {\em output} of $\calT$ on $w$, denoted by $\calT(w)$, is the word $\lambda(\rho)\cdot \omicron(q)$ where $\lambda(\rho)= \lambda(\delta_1)\cdots\lambda(\delta_n)$ and $q$ is accepting state reached by the run. Let $L(\calA)$ denote the set of words accepted by $\calA$, called the \emph{language} of $\calA$ or the \emph{domain} of $\calT$ (denoted as $\dom(\calT)$). We can see that $\calT$ defines a function from $\dom(\calT)$ to $B^*$. Functions defined by sequential transducers are called {\em sequential}. In the literature, they are known as \emph{subsequential functions}, introduced by Sch\"utzenberger \cite{schuetzenberger1977variante}. Transducers given in \Cref{UnboundedFig} are sequential.

If we allow the finite state automaton $\calA$ to be nondeterministic, then $\calT$ no longer defines a function, but a binary relation on $A^{*} \times B^{*}$. Such relations are called \emph{rational}. If the relation is a function, then the transducer is called \emph{functional}, and the corresponding functions are called \emph{rational functions}. 
We can restrict the nondeterminism and still compute all rational functions. A finite state automaton is \emph{unambiguous} if on each input word the machine has at most one run. It is a well-known fact in the theory of transducers that all \emph{rational functions}  are computed by finite state transducers whose underlying automata are unambiguous \cite{choffrut1999uniformization}. Such transducers are called \emph{unambiguous transducers}.  Clearly sequential functions are a strict subset of rational functions. 
For instance, the function `output the input word if the last letter of the input is an $a$, otherwise the empty word' is rational but not sequential.

There exist generalisations of rational functions where the underlying automaton is a two-way finite state automaton or equivalently a finite state automaton with registers (corresponding functions are called \emph{regular} \cite{Engelfriet,SST}), or two-way finite state automaton with pebbles (polyregular functions \cite{bojanczyk2022transducers, Polyregular1}). An overview of the classical theory of transducers is given in \cite{FilliotSigLog}. In this paper, we restrict our attention to one-way functional transducers.

\subsection{Metric on Words, Edit Distances}
Simply put, a metric on a set is used to measure distance between any two elements of the set.  A \emph{metric on words} over the alphabet $A$ is a function $d: A^* \times A^* \rightarrow [0, \infty ]$ such that for any words $u,v$ and $w$ in $A^*$, $d(u,v) = 0 \iff u = v$ (\emph{separation}), $d(u,v) = d(v,u)$ (\emph{symmetry}), and $d(u,v) \leq d(u,w) + d(w,v)$ (\emph{triangle inequality}).

A metric is \emph{integer-valued} if it has range $\N \cup \{\infty\}$.
A trivial metric on words is the \emph{discrete metric} --- distance between words $u$ and $v$, denoted by $d_\infty(u,v)$, is $0$ if $u = v$ and $\infty$ otherwise. Another straightforward distance on words is the absolute difference of their lengths (denoted as $\dlen$). This is a \emph{pseudo-metric} since the distance between two distinct words can be zero, i.e., does not satisfy the separation property of a metric.

An important class of metrics in the context of word transducers is \emph{edit distances}. Loosely speaking, \emph{edits} are operations that transform words, such as \emph{inserting a letter}, \emph{deleting a letter}, \emph{substitutions (letter-to-letter)}, \emph{adjacent transpositions} (swapping adjacent letters), \emph{left and right shifts} etc. 
For a fixed set of edit operations $C$, the edit distance with respect to $C$ between words $u$ and $v$, is the minimum number of edits in $C$ required to transform $u$ to $v$ if it is possible, and $\infty$ otherwise.   
The common edit distances and their corresponding operations are recalled in Table~\ref{table:editdistances}.
\begin{table}[t]
\begin{tabular}{| l|c|p{0.5\textwidth}| } 
 \hline
 Edit Distance & Denotation & Allowed Operations \\ 
\hline
Hamming distance & $d_h$ & letter-to-letter substitutions \\
Transposition distance & $d_t$ & swapping adjacent letters \\
Conjugacy distance & $d_c$ & left and right cyclic shifts \\ 
Levenshtein edit distance & $d_l$ & insertions, deletions, and substitutions \\ 
Longest Common Subsequence & $\dlcs$ & insertions and deletions \\
Damerau-Levenshtein distance  & $\ddl$ & insertions, deletions, substitutions and adjacent transpositions\\
 \hline
\end{tabular}
\caption{Edit Distances}
\label{table:editdistances}
\end{table}
 Since many of these operations are obtained by combinations of the others, we can relate these metrics. The notation $d_1 \leq d_2$ is an abbreviation for $d_1(u,v) \leq d_2(u,v)$ for all words $u,v$.
We can also relate the metrics up to boundedness (See \cite{colcombet2013regular} for a detailed introduction). Let $\alpha:\N \rightarrow \N$ be a \emph{correction} function. Usual examples are increments (e.g.~$x\mapsto x+2$), scaling (e.g.~$x \mapsto 2\cdot x$) etc. We extend $\alpha$ to the domain $\N\cup \{\infty\}$ by letting $\alpha(\infty)=\infty$.  
We write $d_1 \lesssim d_2$ to mean that there is some $\alpha$ such that $d_1 \leq \alpha \circ d_2$. Clearly, if $d_1 \leq d_2$ then $d_1 \lesssim d_2$. If $d_1 \lesssim d_2$ and $d_2 \lesssim d_1$, we write $d_1 \approx d_2$ (this is known as the cost equivalence or the boundedness equivalence). If two functions $f$ and $g$ are cost-equivalent then $f$ and $g$ are bounded over precisely the same family of subsets (See Proposition 1 of \cite{colcombet2013regular}).   

\begin{lemma}
\label{lemma:metricrelation} The metrics defined in Table~\ref{table:editdistances} are related as follows: 
\begin{enumerate}
\item\label{Len:1} $\dlen \leq d \leq d_\infty$, for each edit distance metric $d\in \{d_l,d_h,d_t,d_c,\dlcs,\ddl\}$
\item $d_l \approx \dlcs \approx \ddl$
\item $d_l\leq d_h \lesssim d_t$
\item $d_l \lesssim d_c$
\item $d_c$ and $d_t$ as well as $d_c$ and $d_h$ are incomparable, i.e., $d_h \not \lesssim d_c, d_c \not \lesssim d_h$ and $d_t \not \lesssim d_c, d_t \not \lesssim d_c$
\end{enumerate}
\end{lemma}
\begin{proof}
\begin{enumerate}
\item From definition of the metrics.
\item Clearly $d_l \leq \dlcs$. Since one substitution can be achieved by an insertion and deletion, $\dlcs \leq 2\cdot d_l$. Hence $d_l \approx \dlcs$. Similarly $\ddl \leq d_l$ and since a transposition is equivalent to two substitutions $d_l \leq 2 \cdot \ddl$. Therefore $d_l \approx \ddl$.
\item $d_l \leq d_h$ is obvious. Since a transposition is equivalent to two substitutions $d_h \leq 2 \cdot d_t$.
\item Since a cyclic shift is achieved by an insertion and a deletion, $d_l \leq 2\cdot d_c$. 
 \item For the last statement, consider the family of words $\{((01)^k,(10)^k) \mid k \geq 0 \}$, $d_h,d_t \rightarrow \infty$, but $d_c$ is $1$. Conversely, for the family of words $\{(10^k1,010^{k-1}1) \mid k \geq 1 \}$, $d_t = 1, d_h=2$, but $d_c$ is $\infty$.
\end{enumerate}
\end{proof}
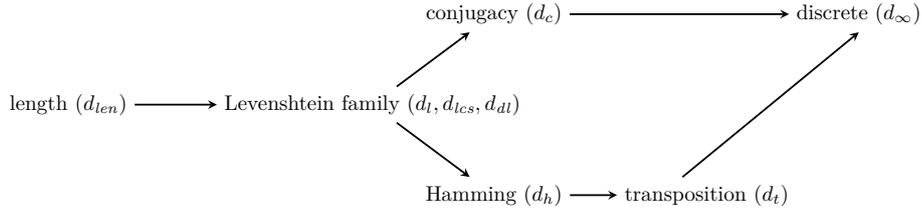
\begin{figure}
\centering\scalebox{.8}{
\begin{tikzpicture}[node distance=1cm]
\node (A) at (-4, 0) {length ($\dlen$)};
\node (B) at (1, 0) {Levenshtein family ($d_l, \dlcs, \ddl$)};
\node (C) at (3, 1.5) {conjugacy ($d_c$)};
\node (D) at (3, -1.5) {Hamming ($d_h$)};
\node (E) at (6.5, -1.5) {transposition ($d_t$)};
\node (F) at (9, 1.5) {discrete ($d_\infty$)};
\draw[-][-stealth][line width=0.3mm](6.1, -1.2) -- (9,1.2);
\draw[-][-stealth][line width=0.3mm]
  (A) edge (B) (B) edge (C) (B) edge (D) (D) edge (E)  (C) edge (F);

\end{tikzpicture}}
\caption{The Boundedness preorder of edit distances}
\label{preorder}
\end{figure}
\section{Distance between Transductions}
\label{sec:3}

In this section, we define the notion of distance between two rational functions, the diameter of a rational relation, and the index of a rational relation in another. We establish the relation between these notions and state our results. 

\subsection{Comparing Transducers}
\label{Subsec:comparing}
We lift a metric over words to the class of word-to-word functions as follows.

\begin{definition}[Metric on  transductions]
\label{defn:transmetric}
Let $d$ be a metric on words over the alphabet $B$.  Given two  partial functions $\calT, \calS : A^\ast \to B^\ast$, we define
$$d(\calT,\calS) = \begin{cases} \sup \left \{\,  d(\calT(w), \calS(w)) \,\mid\,  w \in \dom(\calT) \right \} & \text{ if $\dom(\calT) = \dom(\calS)$} \\
 \infty & \text{ otherwise } 
 \end{cases}$$
\end{definition}

\begin{proposition}
 $d$ is a metric on  transductions.
\end{proposition}

\begin{proof}
Clearly, since $d$ is a metric on words, $d(\calT,\calS)=0$ if and only if $\calT$ and $\calS$ are the same transductions, and $d(\calT,\calS) = d(\calS, \calT)$  for transductions $\calT$ and $\calS$. Let $\calT,\calS,\calF$ be three transductions with the output alphabet $B$. It remains to show that $d(\calT,\calS) \leq d(\calT,\calF) + d(\calF,\calS)$.
  
Assume the domains of $\calT$ and $\calS$ are different. Then, either $\dom(\calT) \neq \dom(\calF)$ or $\dom(\calF) \neq \dom(\calS)$. In both cases, $d(\calT,\calS)$ and 
$d(\calT,\calF) + d(\calF,\calS)$ are $\infty$. Therefore, assume that the domains of $\calT,\calS$ and $\calF$ are the same, call it $L$. Since for each word $w$ in $L$, $d(\calT(w),\calS(w)) \leq d(\calT(w),\calF(w)) + d(\calF(w),\calS(w))$ by virtue of $d$ being a metric, it follows that 
\setlength{\arraycolsep}{0.0em}
\begin{eqnarray*}
d(\calT,\calS) && =\sup\left\{\,d(\calT(w),\calS(w)) \,\mid\, w \in L\,\right\} \\
               && \leq \sup\left\{\,d(\calT(w),\calF(w)) + d(\calF(w),\calS(w))  \,\mid\, w \in L\,\right\}\\
                                          &&\leq \sup\left\{d(\calT(w),\calF(w)) \,\mid\, w \in L \,\right\} + \sup \left\{\, d(\calF(w),\calS(w))  \,\mid\, w \in L \,\right\} \\
                                          &&=  d(\calT,\calF) + d(\calF,\calS)~. \hspace{9cm}\qedhere
\end{eqnarray*}\setlength{\arraycolsep}{5pt}
\end{proof}

\begin{remark}
We can define a notion of distance between word-to-word relations in the above manner, however this distance will not be a metric. In particular $d(R,R)$ will not be $0$ for a relation $R$ that is not a (partial) function.
\end{remark}	

\begin{example}
Consider the sequential transducers $\calT_1$ and $\calT_2$ in \Cref{UnboundedFig}.
The transducers $\calT_1$ and  $\calT_2$ output the letters at the odd and even positions respectively. For any input word $u$, $||\calT_1(u)| - |\calT_2(u)|| \leq 1$. Hence $\dlen(\calT_1,\calT_2)= 1$. For input word $(ab)^n$ where $n > 1$, the outputs produced by $\calT_1$ and $\calT_2$ are $a^n$ and $b^n$ respectively. Since $n$ substitutions are required to convert $a^n$ to $b^n$, $d_l(a^n,b^n) = n$. Therefore, $d_h(\calT_1,\calT_2)=\infty$ as well as $d_l(\calT_1,\calT_2)=\infty$.
\end{example}
 
\begin{example} 
The sequential transducer $\calT_4$ in \Cref{Bounded} replaces each block of $0$'s by a single $0$ and each block of $1$'s by a single $1$. Similarly, $\calT_5$ substitutes a block of $0$'s by a single $1$ and a block of $1$'s by a single $0$. The output words produced by the transducers on any input word is an alternate sequence of $0$'s and $1$'s. If $\calT_4$ outputs $010$, then $\calT_5$ produces its complement, i.e., $101$. The Hamming distance between $\calT_4$ and $\calT_5$ is $\infty$, but the Levenshtein distance is $2$.

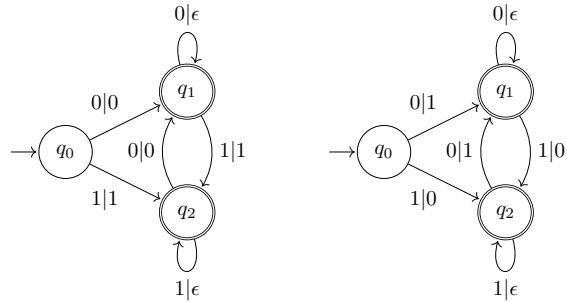
\begin{figure}[htbp]
\centering
\scalebox{.8}{
\begin{tikzpicture}[shorten >=1pt,node distance=3cm,on grid,auto] 
   \node[state,initial,initial text=] at (0,0) (q_0)   {$q_0$}; 
   \node[state,accepting] at (2,1) (q_1) {$q_1$};
   \node[state,accepting] at (2,-1) (q_2) {$q_2$};
    \path[->] 
    (q_0) edge [above left] node {$0|0$} (q_1)
             edge [below left] node {$1|1$} (q_2)
    (q_1) edge [bend left] node  {$1|1$} (q_2)
            edge [loop above] node {$0|\epsilon$} ()
    (q_2) edge [bend left] node  {$0|0$} (q_1)
            edge [loop below] node {$1|\epsilon$} ();
\end{tikzpicture}\hspace{1cm}\begin{tikzpicture}[shorten >=1pt,node distance=3cm,on grid,auto] 
   \node[state,initial,initial text=] at (0,0) (q_0)   {$q_0$}; 
   \node[state,accepting] at (2,1) (q_1) {$q_1$};
   \node[state,accepting] at (2,-1) (q_2) {$q_2$};
    \path[->] 
    (q_0) edge [above left] node {$0|1$} (q_1)
             edge [below left] node {$1|0$} (q_2)
    (q_1) edge [bend left] node  {$1|0$} (q_2)
            edge [loop above] node {$0|\epsilon$} ()
    (q_2) edge [bend left] node  {$0|1$} (q_1)
            edge [loop below] node {$1|\epsilon$} ();
\end{tikzpicture}
}
\caption{$\calT_4$ (left) outputs $0$ \& $1$ for each block of $0$'s \& $1$'s \textit{resp.}~whereas $\calT_5$ (right) outputs $1$ \& $0$ for each block of $0$'s \& $1$'s \textit{resp.}}
\label{Bounded}
\end{figure}
\end{example}

Let $d$ be a distance on words. The value $d(\calT,\calS)$ is an upper bound on how dissimilar the outputs of transducers $\calT$ and $\calS$ can be on any input. It is natural to ask the computational and boundedness problems given in Table~\ref{table:distance}. 
\begin{table}[t]
\begin{tabular}{|p{0.3\textwidth}|p{0.3\textwidth}|p{0.3\textwidth}|} 
 \hline
 Problem & Input & Question \\ 
\hline
Distance Problem &transducers $\calT,\calS$ & $d(\calT,\calS)$? \\
Closeness Problem &transducers $\calT,\calS$ & Is $d(\calT,\calS) < \infty$? \\
$k$-Closeness Problem &integer $k$, transducers $\calT,\calS$ & Is $d(\calT,\calS) \leq k$?  \\ 
 \hline
\end{tabular}
\caption{Problems about distance between two transducers w.r.t.~the metric $d$}
\label{table:distance}
\end{table}

Closeness and $k$-closeness are respectively a boundedness and an upper bound problem on distance. 
\begin{proposition}\label{prop:computedistance}
Let $d$ be an integer-valued metric. The distance problem w.r.t.~$d$ is computable if and only if $k$-closeness and closeness problems w.r.t.~$d$ are decidable.
\end{proposition}
\begin{proof}
Clearly, if we can compute the distance w.r.t.~$d$ then we can decide $k$-closeness as well as closeness. For the other direction, given two transducers, we first check if they are close and if it is we perform an exponential search --- check if they are $k$-close for $k=2^0, 2^1, 2^2,\ldots$ till it fails and subsequently perform a binary search on the interval $[2^n, 2^{n+1}]$, $n\in \N$ that contains the distance. 
\end{proof}
We say two transducers $\calT$ and $\calS$ are \emph{close} (\textit{resp}.~\emph{$k$-close}, for $k \geq 0$) w.r.t.~$d$ if $d(\calT,\calS) < \infty$ (\textit{resp}.~$d(\calT,\calS) \leq k$).
Closeness with respect to the discrete metric $d_\infty$ is precisely the {equivalence problem}. 
In the case of edit distances, closeness means that the output of $\calT_1$ can be converted to the output of $\calT_2$ by doing a bounded number of edits.
\begin{remark}\label{remark:metricrelation}
From \Cref{defn:transmetric}, it is easy to verify that \Cref{lemma:metricrelation} holds for transducers as well. If $d_1 \lesssim d_2$, then it is easy to see that if transducers $\calT_1$ and $\calT_2$ are \emph{not} close w.r.t.~$d_1$, then they are not close w.r.t.~$d_2$ either.
\end{remark}
The problems in Table~\ref{table:distance} for unambiguous transducers with identical domains can be reduced to that for sequential transducers by considering the cartesian product of the unambiguous transducers.  Given two unambiguous transducers $\calT_1$ and $\calT_2$, we obtain the sequential transducers $\calT'_1$ and $\calT'_2$ as follows. The  input automata for $\calT'_1$ and $\calT'_2$ are the same, call it $\calA$, which is the cartesian product of the input automata of $\calT_1$ and $\calT_2$. 
By treating the transitions of the cartesian product as the input alphabet, we get input determinism. The output functions of $\calT'_1$ and $\calT'_2$ are lifted from $T_1$ and $T_2$ respectively.
\begin{proposition}
\label{Prop:rattoseq}
Let $d$ be a distance on words. For each pair of unambiguous transducers $\calT_1$ and $\calT_2$ with identical domain, there exist a DFA $\calA$ and output functions $\lambda_1',\omicron_1'$ and $\lambda_2',\omicron_2'$ such that  $d(\calT_1,\calT_2) = d(\calT_1',\calT_2')$ where the sequential transducer $\calT_i' = \langle \calA, \lambda_i', \omicron_i' \rangle$, $i\in \{1,2\}$. Furthermore, the size of the automaton $\calA$ is polynomial in the size of $\calT_1$ and $\calT_2$. 
\end{proposition}
\begin{proof}
Assume we are given two unambiguous transducers $\calT_1$ and $\calT_2$ with the input alphabet $A$ and the output alphabet $B$. Since they are unambiguous, they have at most one accepting run on any input word. The idea is to construct a deterministic automaton whose language is the set of pairs of successful runs of $\calT_1$ and $\calT_2$ on words in their domain.

For convenience, we can assume that $\calT_1$ and $\calT_2$ have exactly one initial state each, denoted by $i_1$ and $i_2$. 
Let $\Delta_i$, $Q_i$ and $F_i$, $i\in \{1,2\}$, be the transitions, the set of states, and the set of final states of $\calT_i$. 

The input alphabet of $\calA$ is going to be $A'=\Delta_1 \times \Delta_2$.  The deterministic automaton $\calA$ has the set of states $Q_1\times Q_2$, initial state $(i_1,i_2)$, set of final states $F_1\times F_2$, and the set of transitions (denoted as $\Delta$) of the form $((p,q), (\delta,\chi), (p',q'))$ such that $\delta \in \Delta_1,\chi \in \Delta_2$ are of the form $(p,a,p')$ and
$(q,a,q')$ for some letter $a\in A$. It is easy to verify that the language accepted by $\calA$,  denoted as $L' \subseteq (\Delta_1 \times \Delta_2)^*$,  is the language
of all sequences of pairs of the form $\rho=(\delta_1,\chi_1) \cdots (\delta_k,\chi_k), k\geq 1$ and $\delta_i \in \Delta_1$, $\chi_i \in \Delta_2$ for each $1 \leq i\leq k$, such that
\begin{enumerate}
\item $\delta_i$ and $\chi_i$ are transitions on the same letter,
\item $\delta_1 \cdots \delta_k$ (denoted as $\pi_1(\rho)$) is a successful run of $\calT_1$,
\item $\chi_1 \cdots \chi_k$ (denoted as $\pi_2(\rho)$) is a successful run of $\calT_2$.
\end{enumerate}

Let $\lambda_i:\Delta_i \rightarrow B^*, \omicron_i:F_i \rightarrow B^*, i\in \{1,2\}$, be the output functions of $\calT_i$. We define $\lambda_1':\Delta \rightarrow B^*$ and $\lambda_2':\Delta \rightarrow B^*$ as follows: If $\Delta \ni \psi = ((p,q),(\delta,\chi)(p',q'))$, we let $\lambda_1'(\psi) = \lambda_1(\delta)$ and $\lambda_2'(\psi) = \lambda_2(\chi)$. Similarly, we define $\omicron_1',\omicron_2':F_1\times F_2 \rightarrow B^*$ as $\omicron_1'(p,q) = \omicron_1(p)$ and $\omicron_2'(p,q) = \omicron_2(q)$, for each $(p,q) \in F_1\times F_2$. 

We take $\calT_1'$ and $\calT_2'$ to be the automaton $\calA$ with the output functions $\lambda_1',\omicron_1'$ and $\lambda_2'$, $\omicron_2'$ respectively. Clearly, they are of polynomial size w.r.t.~$\calT_1$ and $\calT_2$. It is easy to show that there is a correspondence ($\mapsto$) between $L$ and $L'$, namely $L \ni w \mapsto \rho \in L'$ such that $\pi_1(\rho)$ and $\pi_2(\rho)$ are the unique successful runs of $\calT_1$ and $\calT_2$ on $w$ respectively. Moreover, if $w\mapsto \rho$ then $\calT_1(w) = \calT_1'(\rho)$ and $\calT_2(w) = \calT_2'(\rho)$ ($\star$).

It remains to show that $d(\calT_1,\calT_2) = d(\calT_1',\calT_2')$. Using ($\star$) above, we observe that $d(\calT_1(w),\calT_2(w)) = d(\calT_1'(\rho),\calT_2'(\rho))$ such that $w\in L, \rho \in L'$ and $w \mapsto \rho$. Since $\mapsto$ is a correspondence, we conclude that  $d(\calT_1,\calT_2) = d(\calT_1',\calT_2')$. 
\end{proof}

Using the above result, the closeness w.r.t.~the length metric $\dlen$ can be characterised in terms of delay as follows.
\begin{proposition}\label{Prop:finitedelay}
Consider the sequential transducers $\calT_1, \calT_2$ defined by the DFA $\calA$ and output functions $\lambda_1,\omicron_1$ and $\lambda_2,\omicron_2$ respectively. The following are equivalent.
\begin{enumerate}
\item $\dlen(\calT_1,\calT_2)$ is finite.
\item There is $k\in \N$ such that on any $w\in L(\calA)$, the difference in lengths of the partial outputs of $\calT_1, \calT_2$ on any prefix of $w$ is bounded by $k$.
\end{enumerate}
\end{proposition}
\begin{proof}
$(2) \implies (1)$ is obvious. 
For $(1) \implies (2)$, WLOG assume that automaton $\calA$ is trimmed, i.e., all states are accessible (reachable from the initial state) and coaccessible (from each state there is a path to some final state). Let $\ell$ be the maximum difference in the output lengths on any single transition, and $n$ be the number of states of $\calA$.  We claim that, if the difference in output lengths of the transducers on any input is bounded, then $ k = n \ell$ validates our proposition.  Note that, the difference in the output lengths on any partial input $u$ is at most $ |u|\ell$ and we show that this does not exceed $k$. Suppose $|u|\ell > k = n \ell$, then $|u| > n$, and hence there is a state that is repeated on $u$ such that the repeated part has nonzero difference in the output lengths. That is, $u = u_1u_2u_3 $, with $u_2 \neq \epsilon$ such that the states reached after $u_1$ and that reached after $u_1u_2$ are the same, say $q$, and $|\lambda_1(\rho)| - |\lambda_2(\rho)| \neq 0$, where $\rho$ is the unique partial run from $q$ on $u_2$. If $w = uv$ for some $w \in L(\calA)$ then $u_1u_2^iu_3v \in L(\calA)$ for all $i$, and the difference in the output lengths become unbounded with increasing $i$. This contradicts that $\dlen(\calT_1,\calT_2)$ is finite. Therefore, the difference in the partial outputs of $\calT_1, \calT_2$ on any prefix of $w$ is bounded by $k = n \ell$.
\end{proof}

\newcommand{\dist}[3]{f_{#1, #2}^{#3}}
\newcommand{\kdist}[4]{\lceil \dist{#1}{#2}{#3}\rceil^{\leq #4}}

Given two transductions $\calT$ and $\calS$ we define a distance function which maps each word $w$ to the  distance between their outputs on $w$.
\begin{definition}[Distance function]
	The \emph{distance function} $\dist \calT \calS d :A^*\rightarrow \N \cup \{\infty\}$ of $\calT$ and $\calS$ is  
	$$\dist \calT \calS d (w) = \begin{cases} d(\calT(w),\calS(w))
		& \text{ if $w\in \dom(\calT) \cap \dom(\calS)$} \\
		\infty & \text{ otherwise }
	\end{cases}$$
\end{definition}

Transducers $\calT$ and $\calS$ are close w.r.t.~a metric $d$ if their domains are the same and their distance function $\dist \calT \calS d$ is \emph{limited} (i.e.,$<\infty$ on its domain). Similarly $k$-closeness w.r.t.~$d$ of $\calT$ and $\calS$ reduces to $k$-limitedness of $\dist \calT \calS d$. Limitedness problems are well-studied in the context of weighted automata \cite{leung2004limitedness,colcombet2007factorisation}. Therefore, when the distance function $\dist \calT \calS d$ is computable by a $(\min, +)$-automaton, the distance between $\calT$ and $\calS$ is computable due to \Cref{prop:computedistance}.

However, there are distance functions that are not computable by weighted automata. Let $A=\{a,b\}$. Consider the sequential transducers $\calT_1,\calT_2:A^*\rightarrow A^*$ with the domain $a^*b^*$ defining the functions $a^pb^q\mapsto a^p$, $a^pb^q\mapsto a^q$ respectively ($\calT_1$ outputs the $a$'s and erases the $b$'s, $\calT_2$ erases $a$'s and renames the $b$'s as $a$'s). It is easily checked that their distance function w.r.t.~the Levenshtein family ($d \in \{d_l,\dlcs,\ddl \}$) is $\dist{\calT_1}{\calT_2}{d} : a^pb^q \mapsto |p-q|$.  

If $f:A^*\rightarrow \N\cup\{\infty\}$ is a function computed by weighted automata ($(\min, +)$ or $(\max,+)$ or $B$-automata \cite{colcombet2009theory}), then $L_{f\leq k} = \{ w\in A^* \mid f(w) \leq k\}$ is regular for each $k\leq \N$. Hence the function $\dist{\calT_1}{\calT_2}{d}$ is not realised by any of them (consider the language $L_{\dist{\calT_1}{\calT_2}{d}\leq k}$). In fact, it can be shown that the function $\dist{\calT_1}{\calT_2}{d}$ is not computed even upto boundedness \cite{costfunctions}. 

To compute $k$-closeness w.r.t.~any of the edit distances, it is not necessary to compute the distance function precisely.  The $k$-approximation of the distance function $\dist \calT \calS d$ is the function $\kdist \calT \calS d k : w \mapsto \dist \calT \calS d(w)$ if $\dist \calT \calS d(w) \leq k$ and $\infty$ otherwise. 
\begin{proposition}\label{prop:kapprox}
If $\calT$ and $\calS$ are close w.r.t.~the length metric, then the approximation $\kdist \calT \calS d k$ for a metric $d \in \{d_l,\dlcs,\ddl,d_h,d_t,d_c\}$ is computed by a distance automaton for each $k \in \N$.
\end{proposition}

\begin{proof}
Let $\calT_1$ and $\calT_2$ be two transducers. The domains of these transducers are identical since they are close w.r.t.~the length metric. Hence, we can construct $\calT_1 = \langle \calA, \lambda_1, \omicron_1 \rangle$ and $\calT_2 = \langle \calA, \lambda_2, \omicron_2 \rangle$ with domain $L$ (by virtue of \Cref{Prop:rattoseq}).
Since the transducers  $\calT_1$ and $\calT_2$ are close w.r.t.~the length metric, by virtue of \Cref{Prop:finitedelay}, there is a maximum delay $\partial \in \mathbb{N}$ between any partial outputs of $\calT$ and $\calS$ on any input.

For all metrics $d\in \{d_l,d_h,d_t,\dlcs,\ddl\}$, fix a set of edits $C$ supported in the metric.  We construct a distance automaton $\calA_{C,k}$ --- weighted automaton over the semiring $(\mathbb{N},\min,+, \infty,0)$ ---  that reads an input word $w$ and accepts it if we can perform at most $k$ edits in $C$ to $\calT_1(w)$ and obtain $\calT_2(w)$.  The cost of a run is the number of edits that it uses. For this the states of automaton $\calA_{C,k}$ are those of $\calA$ augmented with a budget $b$ and unmatched leftover word of the form  $(u, \epsilon)$ or $(\epsilon, u)$ with $|u|\leq \mathit{max}(\partial,k)$. Initially, the budget is $k$, and the word is $(\epsilon, \epsilon)$.  While performing a transition $\delta$ of $\calA$, from a state annotated with $b$ and $(\epsilon, u)$, the automaton $\calA_{C,k}$ nondeterministically performs $b'$ edits from $C$ so as to partially match $(\lambda_1(\delta), u\lambda_2(\delta))$. It then moves on to the target state with the reduced budget $b-b'$ and the updated leftover word.  The cost of the transition is the number of edits used in the transition, i.e., $b'$. It is easy to verify that $\calA_{C,k}$ exactly computes $\kdist{\calT_1}{\calT_2}{d}{k}$. 

For conjugacy distance, we construct a new distance automaton $\calA_k$ that remembers a prefix of length at most $k$ of the output of the first (or second) transducer and matches the shifted output with the remaining output of the transducer. The cost of a run is the length of the prefix, i.e., the number of cyclic shifts required to match the outputs.  
 For this, the states of automaton $\calA_k$ are those of $\calA$ augmented with a word of the form $(u, v)$ with $||u| - |v|| \leq \partial$, and  either $|u|\leq k$ or $|v| \leq k$ . Here, $u$ (\textit{resp.}~$v$) is a suitable candidate for the prefix of the output of the first (\textit{resp.}~second) transducer to be matched with the suffix of the output of the second (\textit{resp.}~first) transducer. Initially the word is $(\epsilon, \epsilon)$. While performing a transition $\delta$ of $\calA$, from a state annotated with $(u,v)$, the automaton $\calA_{k}$ either choose to perform step 1 (stores the output of first and second transducer of at most $k$ length), or until nondeterministically chooses to perform step 2 where it starts matching the output of first and second transducer and finally check if the stored output matches with the remaining output.
\begin{enumerate}
\item Go the the target state with updated word as $(u', v') = (u\lambda_1(\delta), v\lambda_2(\delta))$ only if $||u'| - |v'|| \leq \partial$, and $|u'| \leq k$ or $|v'| \leq k$ with the cost of the transition being zero. 
\item There are two symmetric cases. Start to partially match the output of the first (\textit{resp.}~second) transducer $u\lambda_1(\delta)$ (\textit{resp.}~$v\lambda_2(\delta)$) with that of the upcoming output of the second  (\textit{resp.}~first) transducer, i.e., $\lambda_2(\delta)$ (\textit{resp.}~ $\lambda_1(\delta)$) only if $|v| \leq k$ (\textit{resp.}~$|u| \leq k$). It then moves on to the target state with updated word $(\epsilon,v)$ (\textit{resp.}~$(u,\epsilon)$) and also the leftover word obtained after matching $u\lambda_1(\delta)$ and $\lambda_2(\delta)$ (\textit{resp.}~$v\lambda_2(\delta)$ and $\lambda_1(\delta)$). The cost of the transition is set to be $|v|$ (\textit{resp.}~$|u|$). For the upcoming transitions on $\calA$, the matching continues until the output of the second (\textit{resp.}~first) transducer is exhausted with cost of transitions being zero. While matching, the state may need to additionally store the leftover words of the form $(w, \epsilon)$ or $(\epsilon, w)$. After that, it checks if the stored word $v$ (\textit{resp.}~$u$) matches the first (\textit{resp.}~second) transducer's remaining output. 
\end{enumerate}

Note that the size of $\calA_{C,k}$ and $\calA_k$ is exponential in $k$ (as it has to keep track of the unprocessed words).
\end{proof}

To check if $\calT$ and $\calS$ are $k$-close, we check if they have the same domain and they are close w.r.t.~the length metric (otherwise they are neither close nor $k$-close). If so, we check if the domain of $T$ is same as the domain of $\kdist \calT \calS d k$. Thus we get:
\begin{corollary}\label{corr:kclose}
	Let $d$ be any metric from Table~\ref{table:editdistances}, and   $\calT$ and $\calS$ be any functional transducers.  It is decidable if $\calT$ and $\calS$ are  $k$-close w.r.t. $d$.  
\end{corollary}

 Interestingly, as elaborated in \Cref{kclosehammingtrans},  for Hamming and transposition distances we have direct procedures for $k$-closeness that does not involve the weighted automata. 

\subsection{Diameter of a Rational Relation} 
\label{subsec:diameter}

\begin{definition}[Diameter of a Rational Relation w.r.t.~a distance $d$]
The diameter of a rational relation $R$ with respect to a distance $d$, denoted by $\dia{R}{d}$, is the supremum of the distance of the related words in $R$.
$$\dia{R}{d} = \sup \left \{\,  d(u,v) \,\mid\, (u,v) \in R \,\right \}$$
\end{definition}
Similar to the questions asked in Table~\ref{table:distance}, we can ask the questions given in Table~\ref{table:diameter} about diameter of a rational w.r.t.~a metric $d$.
\begin{table}[t]
\begin{tabular}{|p{0.3\textwidth}|p{0.3\textwidth}|p{0.3\textwidth}|} 
 \hline
 Problem & Input & Question \\ 
\hline
Diameter Problem &rational relation $R$ & $\dia{R}{d}$? \\
Bounded Diameter Problem &rational relation $R$& Is $\dia{R}{d} < \infty$? \\
$k$-Bounded Diameter Problem &integer $k$, rational relation $R$ & Is $\dia{R}{d} \leq k$?  \\ 
 \hline
\end{tabular}
\caption{Problems about diameter of a rational relation w.r.t.~the metric $d$}
\label{table:diameter}
\end{table}
We say a rational relation has \emph{bounded (resp.~$k$-bounded) diameter w.r.t.~a distance $d$} if the diameter of the relation w.r.t.~$d$ is finite (\textit{resp.}~$\leq k$).
A rational relation with \emph{bounded delay} are precisely those relations with bounded diameter w.r.t.~a length metric. Relations with 0-delay are called \emph{length-preserving relations} \cite{eilenberg1974automata}  where any two related words are of equal length. It is decidable to check if a rational relation has bounded delay or 0-delay \cite{frougny1991rational}.

Relations bounded w.r.t.~the discrete metric are simply those with only identical pairs. It is decidable to determine if a rational relation $R$ is identity. First, check if $R$ is length-preserving. If so, we can construct a letter-to-letter transducer for $R$ based on Eilenberg and Sch\"utzenberger's theorem \cite{eilenberg1974automata} stating that a length-preserving rational relation over $A^* \times B^*$ is a rational subset of $(A \times B)^*$, or equivalently, it can be realised by a letter-to-letter transducer whose transitions are labelled with elements of $A \times B$.  Finally, validate this transducer for identity by examining the labels of each transition.

\subsection{Index of a Rational Relation in a Composition Closure}
\label{subsec:index}

Consider two rational relations $R$ over $A^* \times B^*$ and $S$ over $B^* \times C^*$. The composition $S \circ R$ over $A^* \times C^*$ is defined by
$(S \circ R)(u) = S(R(u)) = \bigcup_{v \in R(u)} S(v)$.
\begin{definition}[Composition closure of a Rational Relation]
Let $S$ be a rational relation over $A^* \times A^*$. Let $S^{\circn n}$ denote the composition of $S$ with itself $n\geq 0$ times ($S^{\circn 0}$ is taken to be the identity relation), and let $S^{\leq \circn n}$ denotes the composition of $S$ with itself at most $n$ times, i.e., $S^{\leq \circn n} =  S^{\circn 0} \cup S^{\circn 1} \cup \cdots S^{\circn n}$.

The \emph{composition closure} of $S$, denoted as $S^{\circn *}$, is defined as
$S^{\circn *} = \bigcup_{i \geq 0} S^{\circn i}$. 
\end{definition}
Notice that we use parenthesis around the superscript to indicate that the base operation is composition, and not concatenation.
\begin{definition}[Index of a Rational Relation in a Composition Closure]
Let $S$ be a rational relation over $A^* \times A^*$. An \emph{index} of a rational relation $R$ in the composition closure of $S$, denoted as $\indx{R}{S}$, is the smallest integer $k$ such that $R$ is contained in $S^{\leq \circn k}$. 
\end{definition}

\begin{example}
Consider a relation $S$ over $\{a,b\}^* \times \{a,b\}^*$ that deletes the first $a$ if exists on any input. Fix an integer $k > 0$ and let $R$ be the relation that deletes the first $k$ $a$'s from the input if exists. The index of $R$ in $S^{\circn *}$ is $k$ since for any input word $u \in A^*$, $R(u) \in S^{\leq \circn k}(u)$.  

Consider another relation $R^\prime$ that deletes all $a$'s from the input. Since $R^\prime(a^{k+1}) \not \in S^{\leq \circn k}(a^{k+1})$  for any $k >0$, the index of $R^\prime$ in $S^{\circn *}$ is $\infty$.
\end{example}

As seen in the case of the distance and diameter problem, we can ask questions in Table~\ref{table:index} about the index of a rational relation in the composition closure of a relation.
\begin{table}[t]
\begin{tabular}{|p{0.347\textwidth}|p{0.34\textwidth}|p{0.207\textwidth}|} 
 \hline
 Problem & Input & Question \\ 
\hline
Index Problem &rational relation $R$, $S$ & $\indx{R}{S}$? \\
Bounded (or Finite) Index Problem &rational relation $R$, $S$& Is $\indx{R}{S} < \infty$? \\
$k$-Bounded Index Problem &integer $k$, rational relation $R$, $S$& Is $\indx{R}{S} \leq k$?  \\ 
 \hline
\end{tabular}
\caption{Problems about the index of a rational relation in the composition closure of another.}
\label{table:index}
\end{table}
We say a rational relation $R$ has \emph{bounded (resp.~$k$-bounded)} index  in the composition closure of a rational relation $S$ if the index of $R$ in $S^{\circn *}$ is finite (\textit{resp.}~$\leq k$).

The following lemma shows that checking the boundedness of the index problem for an arbitrary rational relation is difficult. 
\begin{lemma}\label{lemma:indexgeneral}
It is undecidable to check if a rational relation has a bounded index in the composition closure of an arbitrary rational relation.
\end{lemma}

\begin{proof}
Proof by reducing a version of the halting problem of the Turing machine to the bounded index problem of a rational relation. 

A Turing machine is a  mathematical model of computation consisting of 7-tuple $M = \langle Q,\Sigma,\Gamma,\delta,q_{in},q_{accept},q_{reject} \rangle$ where $Q$ is a finite set of states; $q_{in}, q_{accept}, q_{reject} \in Q$ is the start state, accept state and reject state respectively; $\Sigma,\Gamma$ are the finite input and tape alphabet respectively; and $\delta$ is the transition function, $\delta: Q \setminus \{q_{accept},q_{reject}\} \times \Gamma \rightarrow Q \times \Gamma \times \{L,R\}$ --- given a state and the current tape symbol (pointed by the tape head) the transition writes a new symbol, changes state and moves the tape-head one cell either to the left or to the right. A configuration of a Turing Machine is a snapshot of the current state of the machine. It can be expressed as a word consisting of three components: the tape contents to the left of the tape head, the current state, and the tape contents to the right of the tape head (including the current cell). For example, configuration $x_1x_2 \cdots x_n q y_1y_2 \cdots y_m$, where $x_i,y_i \in \Gamma$, indicates that tape content is $x_1 \cdots x_ny_1 \cdots y_m$ followed by infinite blanks, the tape head is pointing at $y_1$ and the current state is $q$. On an input word $w$, the initial configuration is $q_{in} w$, and a configuration of the form $xq_{accept} y$ and $x q_{reject} y$, $x,y \in \Gamma^*$, are called the accepting and rejecting configuration respectively. Let $\vdash$ represent a relation between two configurations $C_1,C_2$ such that $C_1 \vdash C_2$ only if $C_2$ is a possible next valid configuration. For instance, for a transition $\delta(q,a) = (q^\prime,b,L)$ and current configuration $xqay$, where $x,y \in \Gamma^*, a \in \Gamma$, a next valid configuration would be $xbq^\prime y$. Hence $xqay \vdash xbq^\prime y$. It is easy to verify that $\vdash$ is a rational relation. A computation on a Turing Machine is a sequence of configurations $C_0, C_1 \cdots C_m$ where each configurtion is obtained from the previous one by the relation $\vdash$, i.e.,
$ C_{i-1} \vdash C_i \quad \forall 1 \leq i \leq m $
and $C_0$ is the initial configuration. If the final configuration, i.e., $C_m$, is accepting then the word is accepted by the machine (accepting computation). If instead $C_m$ is rejecting then the word is rejected by the machine (rejecting computation). There can be computations where the input word is neither accepted nor rejected and the Turing machine runs forever.

\textit{Reduction:} First we recall a version of the halting problem for Turing machines.\\
\textit{Input:} a Turing machine $M$ and an input word $x$\\
\textit{Output:} Yes, if $M$ halts on $x$ and erases the tape contents; No, otherwise.

Given input $M$ and $x$, let $q_{in}$ and $q_{accept}$ be the initial and accepting state of $M$ respectively. We can define a singleton relation $R= \{(q_{in} x, q_{accept})\}$. Let $S$ be the $\vdash$ relation over configurations of $M$. Both $R$ and $S$ are rational relations.

$R$ has a bounded index in $S^{\circn{*}}$ if and only if there exists some $k > 0$ such that $q_{accept} \in \ \vdash^{k}\,(q_{in}x)$. This is equivalent to saying that $R$ has a bounded index in $S^{\circn{*}}$ if and only if $M$ halts on $x$ and erases the tape contents. Since the halting problem for the Turing machine is undecidable, we can conclude that it is undecidable to check if a rational relation has a bounded index in the composition closure of a rational relation.
\end{proof}

However, we show that the index problem is decidable w.r.t.~a large class of rational relations defined below.

\begin{definition}[Metrizable Relation]
\label{metrizable}
Let $S$ be a rational relation over $A^* \times A^*$. Let $d_S: A^* \times A^* \rightarrow \N \cup \{\infty\}$ be the distance between two vertices in the graph of $S$, i.e., for any two words $u$ and $v$, $d_S(u,v)$ is the smallest $i$ such that $v \in S^{\circn i}(u)$, and $\infty$ otherwise.

We say $S$ is a $d$-\emph{metrizable relation} for a metric $d$ if $d_S \approx d$. 
\end{definition}

\begin{proposition}\label{d-index}
Let $R$ be a rational relation and $S$ be a $d$-metrizable relation for an integer-valued metric $d$ for which $\dlen \lesssim d$. If boundedness of diameter w.r.t.~$d$ is decidable for a rational relation, then $\indx{R}{S}$ is computable.
\end{proposition}

\begin{proof}
Similar to distance problem, the index problem is computable iff bounded index and $k$-bounded index problem are decidable. 
For a rational relation $R$ and $d$-metrizable relation $S$, we show that $\indx{R}{S} < \infty$ iff $\dia{R}{d} < \infty$ as follows.
\begin{align*}
\dia{R}{d} < \infty &\iff \exists k\in \mathbb{N} \text{ s.t. } \forall (u,v) \in R, d(u,v) \leq k\\
&\iff \exists k' \in \mathbb{N} \text{ s.t. } \forall (u,v)\in R, d_S(u,v) \leq k' && \text{(Since $d_S \approx d$)}\\
                  &\iff \forall (u,v)\in R, v \in S^{\leq \circn{k'}}(u) && \text{(Definition of $d_S$)}\\
                  &\iff \indx{R}{S} < \infty
\end{align*}

Therefore, if the boundedness of diameter w.r.t.~$d$ is decidable for a rational relation, then we can decide if $\indx{R}{S} < \infty$. If so, then it suffices to decide if $\indx{R}{S} \leq k$ for $k =0, 1, \ldots$ and output the smallest $k$ as the index of $R$ in the composition closure of $S$. 

Since $\dia{R}{d} < \infty$ and $\dlen \lesssim d$, the rational relation $R$ has a bounded delay. Similarly, $S$ also has a bounded delay since for all $(u,v) \in S$, $d_S(u,v) = 1 \implies \exists k \in \mathbb{N}$ s.t. $d(u,v) \leq k$ (since $d_S \approx d$) $\implies \exists k' \in \mathbb{N}$ s.t. $\dlen(u,v) \leq k'$ (since $\dlen \lesssim d$). Since $S$ has bounded delay, for any $k \in \mathbb{N}$, $S^{\circn k}$ also has bounded delay.  It is known that emptinesss  and set difference of two rational relations with bounded delay is decidable (Corollary 2 of \cite{frougny1991rational}). For any $k \in \mathbb{N}$, deciding $\indx{R}{S} \leq k$ reduces to checking if $R \subseteq S^{\circn k}$ (or equivalently, $R \setminus S^{\circn k} = \emptyset$), and hence decidable.
\end{proof}
A close and (almost) dual notion is that of a metric that defines a rational relation.
\begin{definition}[Rationalizable Distance]
A distance $d$ on words is \emph{rationalizable} if the relation $S_d=\{(u,v) \mid d(u,v)=1\}$, called the \emph{distance relation} of $d$, is rational. 
\end{definition}

\begin{example}
Consider the hamming distance $d_h$. We can construct a rational relation $S_h = \{(u,v) \mid u \text { and } v \text{ differ only in exactly one position}\}$.  For example, let $A=\{a,b\}$ and $S_h(aba) = \{bba,aaa,abb\}$. For this, construct a transducer that nondeterministically chooses a position and replaces the input letter with other letters in the alphabet.  Similarly, 
the distance relation of the length metric $S_{len} = \{ (u,v) \mid ||u| - |v|| =1\}$ is also rational.
\end{example}
In fact, we have the following result about the rationalizability of edit distances referred in Table~\ref{table:editdistances}.

\begin{proposition}
Every edit distance $d\in \{d_l,d_h,d_t,d_c,\dlcs,\ddl\}$ is rationalizable.
\end{proposition}
\begin{proof}
Let $d$ be an edit distance. Let $C$ be the set of permissible edits in $d$. We can construct a rational relation $S$ such that $S = \{(u,v) \mid d(u,v) = 1\}$. For $d \in \{d_l,d_h,d_t,\dlcs,\ddl\}$, we can construct a transducer that nondeterministically chooses a position and an edit from set $C$ and applies that edit to the chosen position. For example, the rational relation $S_l$ over $A^* \times A^*$ for Levenshtein edit distance is as follows. For all, $u,v \in A^*$, $(u,v) \in S_l$ if and only if there exists $x,y \in A^*, a,b \in A$ with $a \neq b$, such that either one of the following is true.
\begin{enumerate}
\item $u=xy, v=xay$ (insertion)
\item $u=xay, v=xy$ (deletion)
\item $u = xay, v=xby$ (substitution)
\end{enumerate}
For conjugacy distance, $S_c = \{(u,v) \mid \text{either $u=ax$ and $v=xa$, or $u=xa$ and $v=av$}, x \in A^*, a\in A\}$. For this, we can construct a transducer that nondeterministically remembers (or guesses) the first (or last) letter of the input word and output it at the end (or beginning) along with the rest of the input word. 
\end{proof}
 
\subsection{Reductions between Distance, Diameter and Index Problems}
\label{subsec:reductions}
We show that the distance problem of two rational functions is mutually reducible to the diameter problem of a rational relation, which in turn is mutually reducible to the index problem of a rational relation in the composition closure of a rationalizable distance. Thus, their closeness and boundedness problems are also interreducible.

The correspondence between distance and diameter follows from Nivat's theorem:
\begin{theorem} (\cite{nivat1968transductions}) \label{nivat}
Let $A$ and $B$ be alphabets. The following conditions are equivalent.
\begin{enumerate}
\item $R$ is a rational relation over $A^* \times B^*$.
\item There exist an alphabet C, two alphabetic morphisms $\phi : C^* \rightarrow A^*$ and $\psi : C^* \rightarrow B^*$ and a regular language $L \subset C^*$ such that $R = \{(\phi(w),\psi(w)) \mid w \in L\}$
\end{enumerate}
\end{theorem}
From \Cref{Prop:rattoseq} and $(2) \implies (1)$ in the above theorem, it follows that distance of two rational functions reduces to the diameter of a rational relation. Now, given a rational relation $R$, we can create two functional transducers $\calT_1$ and $\calT_2$ in the following way. The domain for these transducers corresponds to the set $L$ in \Cref{nivat}. For each transition in $\calT_1$ and $\calT_2$ that involves an input alphabet symbol $\sigma$, we set the outputs to be $\phi(\sigma)$ and $\psi(\sigma)$ in \Cref{nivat}, respectively. Consequently, $\calT_1$ and $\calT_2$ consist of the sets $\{\phi(w) \mid w \in L\}$ and $\{\psi(w) \mid w \in L\}$ respectively. Since the domain of these transducers is identical, the distance between $\calT_1$ and $\calT_2$ with respect to any distance $d$, 
$d(\calT_1,\calT_2) =  \sup \left \{\,  d(\phi(w),\psi(w)) \,\mid\, w \in L \,\right \}$,
 that is equivalent to the diameter of $R$ w.r.t.~the distance $d$.

The correspondence between diameter and  index  for rationalizable distances is stated in the following proposition.
\begin{proposition}\label{reduction2}
The diameter of a rational relation $R$ w.r.t.~a rationalizable distance $d$ is equal to the index of the rational relation $R$ in the compostion closure of the distance relation of $d$.
\end{proposition}

\begin{proof}
Assume that the diameter of a relation $R$ w.r.t.~a distance $d$ is $\infty$. We claim that the index of $R$ in $S_d^{\circn *}$ is also $\infty$ where $S_d$ is the distance relation of $d$. Suppose not, i.e., let $k<\infty$ be the index of $R$ in $S_d^{\circn *}$. Thus, $\forall (u,v)\in R$, $v\in S_d^{\leq \circn k}(u)$. Since $S_d$ is the distance relation of $d$, $\forall (u,v)\in R$, $d(u,v)\leq k$. However, this contradicts the fact that $\dia{R}{d}=\infty$. Hence, the index of $R$ in $S_d^{\circn *}$ is infinite. Similarly, we can prove the other direction. 
Now, suppose the diameter of $R$ w.r.t.~$d$ is finite, i.e.,
\begin{align*}
\text{diameter of $R$ w.r.t.~$d$ is $k < \infty$} &\iff \forall (u,v) \in R, d(u,v) \leq k\\
                  &\iff \forall (u,v)\in R, v \in S_d^{\leq \circn k}(u)\\
                  &\iff \text{ index of $R$ in $S_d^{\circn *}$ is $k$.}  
                  && \qedhere
\end{align*}
\end{proof}

\subsection{Decidability Results}
We study the problems stated in Tables~\ref{table:distance}, \ref{table:diameter} and \ref{table:index} and show that they are decidable for the metrics in Table~\ref{table:editdistances}. The index problems stated in Table~\ref{table:index} are undecidable in general (see \Cref{lemma:indexgeneral}), but is decidable for $d$-metrizable relations for metrics $d$ given in Table~\ref{table:editdistances}.

Recall that, w.r.t.~a metric $d$, distance problem is computable if and only if both closeness and $k$-closeness are decidable (see \Cref{prop:computedistance}). We have shown that the $k$-closeness is decidable for all the metrics in Table~\ref{table:distance} (\Cref{corr:kclose}). Hence to show the decidability of all the problems in Table~\ref{table:distance}, it suffices to show the decidability of the closeness problem. Furthermore, thanks to the inter-reductions described above (see \Cref{subsec:reductions}), the decidability of Table~\ref{table:diameter} follows as well as the decidability for the problems in Table~\ref{table:index} for the rationalizable distance. Moreover, the index of a rational relation in the composition closure of a $d$-metrizable relation for a metric $d$ given in Table~\ref{table:distance} is computable by \Cref{d-index}.

It only remains to prove that closeness is decidable for edit distances in Table~\ref{table:editdistances}. This is stated below, and proved in the following section.

\begin{theorem}\label{thm:closeDecidable}
Let $d$ be any metric from Table~\ref{table:editdistances}, and  $\calT$ and $\calS$ be any functional transducers.  It is decidable if $\calT$ and $\calS$ are  close w.r.t. $d$.  
\end{theorem}

\section{Closeness for Edit Distances}
\label{Subsec:closeness}

In this section we show that closeness is decidable for all the edit distances in Table~\ref{table:editdistances}. The first step is to check if the domain of the transducers are the same. This reduces to checking the equivalence of the underlying automata. For sequential transducers, the underlying automaton is a DFA, while for unambiguous transducers, the underlying automaton is an unambiguous NFA. Checking the equivalence of two unambiguous automata can be done in polynomial time \cite{SteHun85}, while it is PSPACE in the case of ambiguous automata \cite{stock73}.  Therefore, from now on, we assume that the domains of the transducers given as input to the closeness problem are identical.

\Cref{Prop:rattoseq} allows us to state the distance and closeness problems more abstractly in terms of an automaton over pairs of words. The proposition asserts that distance and closeness problems of two given sequential or unambiguous transducers $\calT_1$ and $\calT_2$ can be reduced to the corresponding problem for a DFA $\calA$ with two sets of output functions $\lambda_1,\omicron_1$ and $\lambda_2,\omicron_2$. We can combine the output functions to output a pair of words. That is to say, let $\lambda:\Delta \rightarrow B^*\times B^*$ be defined as $\lambda(\delta) = (\lambda_1(\delta),\lambda_2(\delta))$, where $\delta \in \Delta$ and $\Delta$ is the set of transitions of $\calA$. Similarly let  $\omicron(p) = (\omicron_1(p),\omicron_2(p))$, where $p \in F$ and $F$ is the set of accepting states of $\calA$. Henceforth, we can assume that we are given a DFA $\calA$ with the output functions $\lambda$ and $\omicron$, denoted as the sequential transducer $\calT$. Since the input words are inconsequential for computing the distance, we can convert the transducer $\calT$ to an automaton $\calA$ that accepts a set of pairs of output words over $B^* \times B^*$, i.e., $L(\calA) = \{(u,v) \in B^* \times B^* \mid (u,v) = \calT(w), w \in \dom(\calT)\}$. Clearly, transducers $\calT_1$ and $\calT_2$ are close w.r.t.~$d$ if and only if there exist an integer $k \geq 0$ such that $\forall (u,v) \in L(\calA)$, $d(u,v) \leq k$. 

Conjugacy of words plays an important role in closeness problems. A pair of words $(u,v)$ is \emph{conjugate} if there exist words $x$ and $y$ (possibly empty) such that $u=xy \text{ and }v=yx$ or equivalently, $u$ and $v$ are cyclic shifts of one another. For example, $(\mathit{ aaab}, \mathit{aaba})$ is a conjugate pair where $x=a$ and $y=aab$. Conjugacy relation is an equivalence relation on the set of words. A set of pairs is \emph{conjugate} if each pair in the set is conjugate.

\begin{lemma}\label{kleeneconjugacy}
Let $\calT_1$ and $\calT_2$ be two sequential transducers that define a function from $A^*$ to $B^*$. If $\calT_1$ and $\calT_2$ are close w.r.t.~a metric $d\in \{d_l,d_h,d_t,d_c,\dlcs,\ddl\}$, then every loop in the trim automaton over $B^* \times B^*$, that accepts set of all pairs of output words of $\calT_1$ and $\calT_2$ on any input, generates only conjugate pair of words.
\end{lemma}

\begin{proof} This proof is an adaptation of a related result in \cite{decidingconjugacy}. 
Let $\calA$ be a trim automaton that realises the pair of output words of transducers $\calT_1$ and $\calT_2$ on any input. Since $\calT_1$ and $\calT_2$ are close w.r.t.~$d$, there exist an integer $k \geq 0$ such that $\forall (u,v) \in L(\calA)$, $d(u,v) \leq k$.  Let $(u,v)$ be a pair labelled in a loop rooted at some state $q$. Hence $(u^\ell,v^\ell)$ for each $\ell \geq 0$ is also a pair in a loop rooted at $q$. We can safely assume that $|u| = |v|$, otherwise the edit distance will be unbounded as each iteration will increase the edit distance by a difference in length of $u$ and $v$  (\Cref{Len:1} of \Cref{lemma:metricrelation}). 

Since $\calA$ is trimmed, there exists a path from an initial state $q_0$ to $q$ and from $q$ to a final state $q_f$. Let $(\as_0,\bs_0)$ be a pair labelled in a path from $q_0$ to $q$, and let $(\as_1,\bs_1)$ be a pair labelled in a path from $q$ to $q_f$. Thus, pair $(\as_0,\bs_0)(u^\ell,v^\ell)(\as_1,\bs_1)$ belongs to $L(\calA)$ where $\ell = 2^k$ (some value much larger than $k$). Since $\ell$ is much larger than $k$ and $d(\as_0u^\ell \as_1,\bs_0v^\ell \bs_1) \leq k$, there exist large portions of $u$'s and $v$'s that match. Therefore, we can infer that $u$ is a factor of $vv$, and $v$ is a factor of $uu$.

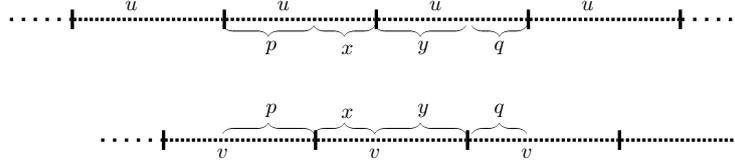
\begin{figure}[t]
\centering
\scalebox{.8}{
\begin{tikzpicture}

\draw[line width=.55mm,loosely dotted] [-](-1,2) -- (0,2)[anchor=south, xshift=-1.5cm, yshift=0cm]node{};
\draw[line width=.55mm,densely dotted] [|-](0,2) -- (2.5,2)[anchor=south, xshift=-1.5cm, yshift=0cm]node{$u$};
\draw[line width=.55mm,densely dotted] [|-](2.5,2) -- (5,2)[anchor=south, xshift=-1.5cm, yshift=0cm]node{$u$};
\draw[line width=.55mm,densely dotted] [|-](5,2) -- (7.5,2)[anchor=south, xshift=-1.5cm, yshift=0cm]node{$u$};
\draw[line width=.55mm,densely dotted] [|-](7.5,2) -- (10,2)[anchor=south, xshift=-1.5cm, yshift=0cm]node{$u$};
\draw[line width=.55mm,loosely dotted] [|-](10,2) -- (11,2)[anchor=south, xshift=-1.5cm, yshift=0cm]node{};

\draw[line width=.55mm,loosely dotted] [-](.5,0) -- (1.5,0)[anchor=north, xshift=-1.5cm, yshift=0cm]node{};
\draw[line width=.55mm,densely dotted] [|-](1.5,0) -- (4,0)[anchor=north, xshift=-1.5cm, yshift=0cm]node{$v$};
\draw[line width=.55mm,densely dotted] [|-](4,0) -- (6.5,0)[anchor=north, xshift=-1.5cm, yshift=0cm]node{$v$};
\draw[line width=.55mm,densely dotted] [|-](6.5,0) -- (9,0)[anchor=north, xshift=-1.5cm, yshift=0cm]node{$v$};
\draw[line width=.55mm,densely dotted] [|-](9,0) -- (11,0)[anchor=north, xshift=-1.5cm, yshift=0cm]node{};

\draw [decorate,
decoration = {calligraphic brace,
raise=3pt,
amplitude= 5pt,mirror
}] (2.5,2) -- (4,2);
\draw [decorate,
decoration = {calligraphic brace,
raise=3pt,
amplitude= 5pt,mirror
}] (4,2) -- (5,2);

\draw [decorate,
decoration = {calligraphic brace,
raise=3pt,
amplitude= 5pt,mirror
}] (5,2) -- (6.5,2);

\draw [decorate,
decoration = {calligraphic brace,
raise=3pt,
amplitude= 5pt,mirror
}] (6.6,2) -- (7.5,2);

\draw [decorate,
decoration = {calligraphic brace,
raise=3pt,
amplitude= 5pt
}] (2.5,0) -- (4,0);
\draw [decorate,
decoration = {calligraphic brace,
raise=3pt,
amplitude= 5pt
}] (4,0) -- (5,0);

\draw [decorate,
decoration = {calligraphic brace,
raise=3pt,
amplitude= 5pt
}] (5,0) -- (6.5,0);

\draw [decorate,
decoration = {calligraphic brace,
raise=3pt,
amplitude= 5pt
}] (6.6,0) -- (7.5,0);

\draw[draw=none] [-](2.5,2) -- (4,2)[anchor=south, xshift=-0.7cm, yshift=-.7cm]node{$p$};
\draw[draw=none] [-](4,2) -- (5,2)[anchor=south, xshift=-0.45cm, yshift=-.7cm]node{$x$};
\draw[draw=none] [-](5,2) -- (6.5,2)[anchor=south, xshift=-0.7cm, yshift=-.7cm]node{$y$};
\draw[draw=none] [-](6.5,2) -- (7.5,2)[anchor=south, xshift=-0.45cm, yshift=-.7cm]node{$q$};

\draw[draw=none] [-](2.5,0) -- (4,0)[anchor=south, xshift=-0.7cm, yshift=.25cm]node{$p$};
\draw[draw=none] [-](4,0) -- (5,0)[anchor=south, xshift=-0.45cm, yshift=.25cm]node{$x$};
\draw[draw=none] [-](5,0) -- (6.5,0)[anchor=south, xshift=-0.7cm, yshift=.25cm]node{$y$};
\draw[draw=none] [-](6.5,0) -- (7.5,0)[anchor=south, xshift=-0.45cm, yshift=.25cm]node{$q$};

\end{tikzpicture}}
\caption{$v$ as infix of $uu$.}

\label{uvlarge}
\end{figure}
Since $v$ is an infix of $uu$, the following holds as shown in \Cref{uvlarge}. There exist words $x,y,p$ and $q$ such that $v = xy$ and $u = px = yq$. 
Since $|u| = |v|$, length of $p$ and length of $y$ are the same, that implies $p=y$ (since $u = px = yq)$. Therefore, $u = yx$. Hence $u$ and $v$ are conjugate words. Since the pair $(u,v)$ was arbitrary, any pair generated by a loop in $\calA$ is conjugate.
\end{proof}

\subsection{Closeness w.r.t.~Levenshtein distances and Conjugacy}
In this subsection, we decide closeness w.r.t.~Levenshtein family of distances --- Levenshtein, Damerau-Levenshtein, and LCS distances --- and conjugacy distance. Levenshtein family of distances are all equivalent with respect to closeness problems by \Cref{lemma:metricrelation} and \Cref{remark:metricrelation}.

We have already seen that given two unambiguous transducers $\calT_1$ and $\calT_2$ with identical domains, there exists an automaton $\calA$ over $B^* \times B^*$ that accepts a set of all pairs of output words of $\calT_1$ and $\calT_2$ on any input. Thus, we can state the distance and closeness problems in terms of rational expressions over $B^* \times B^*$.

We define \emph{pairs over the alphabet $B$} to be the set $B^* \times B^*$ with the pointwise concatenation $(u,v) \cdot (u',v') = (u\cdot u', v \cdot v')$. 
A \emph{rational expression of pairs} over the alphabet $B$ is a rational expression over the alphabet $\{(b,b') \mid b,b' \in \(B \cup \{\epsilon\}\)\}$ that generates subset of pairs over $B$.

From the automaton $\calA$ over $B^* \times B^*$,  using state elimination method (\cite{kozen}, Lecture 9), we can construct the rational expression of pairs $E$  for the output pairs generated by the transducer $\calT_1$ and $\calT_2$ on any input. 
We can lift the metric $d$ to expressions by letting $d(E) = \sup{\{d(u,v) \mid (u,v) \in L(E)\}}$.
Clearly $d(E) = d(\calT_1,\calT_2)$. Thus, the distance and closeness problems of sequential and unambiguous transducers reduce to the corresponding problems for a rational expression of pairs. Henceforth we assume that we are given a rational expression of pairs. 

In the context of conjugacy distance, the closeness of a rational expression necessarily implies that every pair in the expression is conjugate. Otherwise, if there exists a pair $(u,v) \in L(E)$ such that $u$ is not conjugate to $v$, then $d_c(u,v) = \infty$, thus $d_c(E) = \infty$. In fact, this is also a sufficient condtion. The proof relies on the results from \cite{decidingconjugacy} that studies the conjugacy of rational expression over pairs of words. 
It crucially uses the notion of a \emph{common witness} of a set of pairs.
\begin{definition}[Common Witness of a Set of Pairs]
 A \emph{witness} of  pair of conjugate words $(u,v)$ is a word $z$ such that either $uz = zv$ (called an \emph{inner witness}) or $zu = vz$ (called an \emph{outer witness}). A \emph{common witness} of a set of pairs is a word $z$ such that either $z$ is an inner witness of every pair in the set, or $z$ is an outer witness of every pair in the set.
\end{definition}
Lyndon and Sch\"utzenberger in 1962 gave a characterisation of conjugacy of a pair of words, stated as a pair of words is conjugate if and only if it has both inner and outer witness (Proposition 1.3.4 of \cite{lyndonSch}). In \cite{decidingconjugacy}, it is generalised to a set of pairs as follows.
\begin{theorem} (\cite{decidingconjugacy})\label{conjugacytheorem}
Let $M = (\alpha_0,\beta_0) G_1^* (\alpha_1,\beta_1) \cdots G_k^*(\alpha_k,\beta_k)$ be a set of pairs where $G_1, \ldots, G_k$, $k > 0$ are arbitrary sets of pairs of words, and $(\alpha_0,\beta_0), \ldots, (\alpha_k,\beta_k)$ are arbitrary pairs of words. The set $M$ is conjugate iff $M$ has a common witness.
\end{theorem}
Existence of a common witness bounds the conjugacy distance of an expression as follows.
\begin{claim}\label{conjugacyclaim}
If a rational expression over pairs $E$ has a common witness $z$, then $d_c(E) \leq |z|$.
\end{claim}

\begin{proof}
Since $E$ has a common witness, either $\forall (u,v) \in L(E)$, $uz = zv$, or $\forall (u,v) \in L(E)$, $zu = vz$. WLOG, assume that  $\forall (u,v) \in L(E)$, $uz = zv$. Now, for any pair $(u,v) \in L(E)$:
\begin{enumerate}
\item If $|u| > |z|$, then $z$ is a prefix of $u$ and suffix of $v$ and hence $(u,v) = (zp,pz)$ for some word $p \in A^*$. Therefore $d_c(u,v) \leq |z|$ since $v$ can be  obtained by $|z|$ left cyclic shifts of $u$.
\item Otherwise, when $|u| \leq |z|$, the number of cyclic shifts required to transform $u$ to $v$ (note that $u$ and $v$ are conjugate since they have a witness) is less than $|u| \leq |z|$. \qedhere
\end{enumerate}
\end{proof}
A rational expression is \emph{sumfree} if it does not use sum (i.e., $+$). In \cite{decidingconjugacy}, it is shown that if  a common witness exists, it is computable 
for a sumfree rational expression over pairs of words.
It is folklore that every rational expression is equivalent to a sum of sumfree expressions \cite{decidingconjugacy}. 
The proposition below implies that to show closeness for a sum of sumfree expressions, it suffices to show closeness for each of its constituent sumfree expressions. 
\begin{proposition}\label{maxedit}
Let $E = E_1 + \cdots + E_k, k\geq 1$ be a rational expression of pairs. Then $d(E) =  \max (d(E_1), \ldots, d(E_k))$ for all word metrics $d$.
\end{proposition}
\begin{proof}
It suffices to show that if $E=E_1+E_2$, then $d(E) = \max(d(E_1),d(E_2))$. Since $L(E_i) \subseteq L(E), i \in \{1,2\}$, we can deduce that $d(E_i) \leq d(E)$ and hence $\max(d(E_1),d(E_2)) \leq d(E)$. It remains to show that $d(E) \leq \max(d(E_1),d(E_2))$. We have two cases.
\begin{enumerate}
\item If $d(E) = k \in \N$, then for each pair in $E$ the distance is at most $k$, and there is a pair $(u,v) \in L(E)$ with distance $k$. Then, for each pair in $E_1$ as well as $E_2$ the distance is at most $k$, and the expression that contains the pair $(u,v)$ has distance $k$. 

\item If $d(E) = \infty$, then either there is a pair $(u,v)$ with distance $\infty$, in which case one of $d(E_1),d(E_2)$ is $\infty$. Otherwise, there is an infinite subset of pairs $L \subseteq L(E)$ such that for each $k \in N$ there is a pair with distance at least $k$. Since $L$ is infinite, one of $L \cap L(E_1), L \cap L(E_2)$ is infinite, and the corresponding expression has distance $\infty$. \qedhere
\end{enumerate}
\end{proof}
An expression is conjugate, if every pair generated by the expression is conjugate. The following proposition characterises closeness w.r.t.~conjugacy distance.

\begin{proposition}\label{closewrtconjugacy}
A rational expression over pairs of words is close w.r.t.~conjugacy distance if and only if the expression is conjugate. Furthermore, the closeness w.r.t.~conjugacy distance is decidable.
\end{proposition}

\begin{proof}
One direction is trivial. Assume $E$ to be an arbitrary rational expression of pairs and is conjugate. Let $E = E_1 + E_2 +\cdots + E_k$ where $E_1,E_2, \ldots, E_k$ are sumfree expressions. Since $E$ is conjugate, each of its sumfree constituent $E_i$ for $1 \leq i \leq k$ is also conjugate. Using \Cref{conjugacytheorem}, each $E_i$ has a common witness, say $z_i$. From \Cref{conjugacyclaim}, $d_c(E_i) \leq z_i$. Therefore, $d_c(E)$ is close w.r.t.~conjugacy distance by \Cref{maxedit}. Hence, to decide closeness of $E$ w.r.t.~conjugacy distance, it suffices to check if $E$ is conjugate. This reduces to checking if a common witness exists for each sumfree constituents. It is shown to decidable in \cite{decidingconjugacy}.
\end{proof}

Now consider the case of Levenshtein distances. From \Cref{kleeneconjugacy}, if an expression is close w.r.t.~Levenshtein distances, it is necessary that every pair generated by a Kleene star in the expression needs to be conjugate. Using common witness, we show that it is also a sufficient condition.
\begin{claim}\label{closenessClaim}
If a rational expression of pairs $E$ has a common witness $z$, then $d_l(E) \leq 2|z|$.
\end{claim}

\begin{proof}
The proof is similar to \Cref{conjugacyclaim}. Since $E$ has a common witness, either $\forall (u,v) \in L(E)$, $uz = zv$, or $\forall (u,v) \in L(E)$, $zu = vz$. WLOG, assume that  $\forall (u,v) \in L(E)$, $uz = zv$. For any pair $(u,v) \in L(E)$, $|u| = |v|$ since $uz = zv$. There are two cases, either $|u| > |z|$ or $|u| \leq |z|$. If $|u| > |z|$, then $z$ is a prefix of $u$ and suffix of $v$ and hence $(u,v) = (zp,pz)$ for some word $p$. Therefore, $d_l(E) \leq 2|z|$ by deleting $z$ in the beginning and insert $z$ at the end of $u$. Suppose $|u| \leq |z|$, the number of edits required to transform $u$ to $v$ is less than $|u| + |v| \leq 2|u| \leq 2|z|$.
\end{proof}

\begin{proposition}
Closeness of a rational expression w.r.t.~Levenshtein distance is decidable.
\end{proposition}

\begin{proof}
Given an arbitrary rational expression, there is an equivalent sum of sumfree expression. From \Cref{maxedit}, to show closeness for a sum of sumfree expressions, it suffices to show closeness for each of its constituent sumfree expressions. The general form of a sumfree expression $E = (\alpha_0,\beta_0) E_1^*(\alpha_1,\beta_1) \cdots E_k^*(\alpha_k,\beta_k)$ where $k \in \mathbb N$, for $0 \leq j \leq k$, $(\as_j,\bs_j)$ is a (possibly empty) pair of words, and for each $1 \leq i \leq k$, $E_i$ is a sumfree expression.

\begin{claim}\label{closewrtlevenshtein}
A sumfree expression $E = (\alpha_0,\beta_0) E_1^*(\alpha_1,\beta_1) \cdots E_k^*(\alpha_k,\beta_k)$ is close w.r.t.~Levenshtein distance if and only if each $E_i^*$ for $1 \leq i \leq k$ is conjugate.
\end{claim}

\begin{claimproof}
From \Cref{kleeneconjugacy}, if $E$ is close w.r.t.~Levenshtein edit distance then each $E_i^*$ is conjugate. For the other direction, if each $E_i^*$ is conjugate, then each $E_i^*$ has a common witness, say $z_i$, by \Cref{conjugacytheorem}. From \Cref{closenessClaim}, $d_l(E_i^*) \leq 2|z_i|$. Further, $d_l(E) \leq \sum_{j \in \{0 \ldots k\}}d_l(\as_j,\bs_j) + \sum_{i \in \{1 \ldots k\}}d_l(E_i^*) = \sum_{j \in \{0 \ldots k\}}d_l(\as_j,\bs_j) + 2\sum_{i \in \{1 \ldots k\}}z_i$, hence finite. This implies that if each $E_i^*$ in $E$ is conjugate, then $d_l(E)$ is finite.
\end{claimproof}
Therefore, checking the closeness of a rational expression w.r.t.~Levenshtein distances reduces to checking the existence of a common witness of each Kleene star in its sumfree constitutents, and thus decidable.
\end{proof}
For a sumfree rational expression, a witness, if exists, can be computed in polynomial time \cite{decidingconjugacy}, and thus closeness w.r.t.~Levenshtein and conjugacy distances are decidable in polynomial time. However, converting a rational expression to a sum of sumfree rational expressions can cause an exponential blow-up both in the number of summands and the size of each summand  \cite{decidingconjugacy}.

\subsection{Closeness w.r.t.~Hamming and Transposition distances}
\begin{theorem}
Closeness w.r.t.~Hamming and Transposition distance are decidable for functional transducers.
\end{theorem}

Given two functional transducers, check if their domains are the same. If not, the distance is $\infty$ hence they are not close. 
Assume they have an identical domain. By \Cref{Prop:rattoseq}, it suffices to consider two sequential transducers with a common underlying DFA. Let $\calT_1 = \langle \calA, \lambda_1, \omicron_1\rangle$ and $\calT_2 = \langle \calA, \lambda_2, \omicron_2\rangle$ be two sequential transducers. WLOG, we make the following assumptions.
\begin{enumerate}

\item (Property $\star$) Automaton $\calA$ is trimmed, i.e., all states are accessible (reachable from the initial state) and coaccessible (from each state there is a path to some final state).

\item (Property $\ddagger$) $\calT_1$ and $\calT_2$ produce output words of identical length: otherwise the Hamming as well as transposition distance will be $\infty$. We can check this property: rename all the output letters in $\calT_1$ and $\calT_2$ to $a$ and check their equivalence.
\item The delay between partial outputs of $\calT_1$ and $\calT_2$ is at most $k \in \N$ (By \Cref{Prop:finitedelay}).
\end{enumerate}

Let $Q$ and $F\subseteq Q$ be the set of states and final states of $\calA$ respectively, and let $q_0\in Q$ be the initial state. For states $p,q \in Q$, Let $M_{p,q}$ be the set of pairs $(u,v)$ such that there is a run $\rho$ from $p$ to $q$ and $u = \lambda_1(\rho)$ and $v=\lambda_2(\rho)$. Extending this notation, for a state $q_f \in F$, let $M'_{q,q_f}$ be the set of pairs $(u,v)$ such that $u = u'\cdot \omicron_1(q_f), v = v' \cdot \omicron_2(q_f)$ and $(u',v')\in M_{q,q_f}$.

Let $q$ be a state of the automaton. If $(\alpha,\beta)$ and $(\alpha',\beta')$ are two pairs in  $M_{q_0,q}$, then 
$|\alpha|-|\beta| = |\alpha'|-|\beta'|$, or else one of the pairs in $\{ (\alpha\alpha'', \beta\beta''),  (\alpha'\alpha'', \beta'\beta'')\}$ will have different lengths, where $(\alpha'',\beta'')$ is some pair in $M_{q,q_f}$, for some $q_f \in F$,  guaranteed by Property ($\star$). Therefore with each state $q$, we can associate the delay of a run reaching it, called the \emph{delay at $q$}, denoted by $\partial_q$, as $|\alpha|-|\beta|$. Clearly $\partial_q \leq k$. 
By a symmetric argument, if $(\alpha,\beta)$ and $(\alpha',\beta')$ are two pairs in  $M_{q,q_f}$, where $q_f$ is some final state, then $|\alpha|-|\beta| = |\alpha'|-|\beta'| = -\partial_q$.
This also implies that for all $(u,v) \in M_{q,q}$, $|u| = |v|$.

For each state $q$, either $M_{q,q}=\{(\epsilon,\epsilon)\}$, or $M_{q,q}$ is infinite. Let $q$ be a state for which $M_{q,q}$ is nonempty.  For a delay $\partial \in \mathbb{Z}$, a pair $(u,v) \in M_{q,q}$ where $n= |u| > \partial$, we define the interior of the pair $(u,v)$ as 
\begin{align*}
\interior_{\partial}(u,v)&= \begin{cases} (u[1\ldots n-\partial],v[\partial+1\ldots n]) &\text{if $\partial \geq 0$}\\
(u[\partial+1\ldots n],v[1\ldots n-\partial]) &\text{if $\partial < 0$}
\end{cases}
\end{align*}
For example, $\interior_{1}(abc, def) = (ab,ef)$ and $\interior_{-1}(abc, def) = (bc,de)$. We also define the \emph{Left-Border} and \emph{Right-Border} of the pair $(u,v)$ as\\
\begin{minipage}{0.35\textwidth}
\noindent	\begin{align*}
		\lborder_{\partial}(u,v) &= \begin{cases} v[1\ldots \partial] &\text{if $\partial \geq 0$}\\
			u[1\ldots \partial] &\text{if $\partial < 0$}
		\end{cases}
	\end{align*}
\end{minipage} 
\hfill
\begin{minipage}{0.6\textwidth}
\begin{align*}
	\rborder_{\partial}(u,v) &= \begin{cases} u[n-\partial+1\ldots n] &\text{if $\partial \geq 0$}\\
		v[n-\partial+1\ldots n] &\text{if $\partial < 0$}
	\end{cases}
\end{align*}
\end{minipage}

\begin{claim}\label{Claim:hamming}
Hamming distance between $\calT_1$ and $\calT_2$ is unbounded if and only if there exists a state $q\in Q$ and $(u,v) \in M_{q,q}$ such that $|u|=|v| > \partial_q$, and $u' \neq v'$ where $(u',v')=\interior_{\partial_q}(u,v)$.
\end{claim}

\begin{proof}
The \Cref{fig:jaro} depicts the situation described by $(2)$. 

($\leftarrow$): Assume there exists a state $q\in Q$ and $(u,v) \in M_{q,q}$ such that $|u|=|v| > \partial_q$, and $u' \neq v'$ where $(u',v')=\interior_{\partial_q}(u,v)$. Let $(\alpha_0,\beta_0)\in M_{q_0,q}$ and $(\alpha_1,\beta_1)\in M_{q,q_f}$. Consider the pair $(u_i = \alpha_0u^i\alpha_1, v_i = \beta_0v^i\beta_1)$, $i\geq 1$ (shown in \Cref{fig:jaropump}). Since $u'\neq v'$, we can deduce that $d_h(u_i,v_i) \geq i$. Hence $d_h(\calT_1,\calT_2) = \infty$.

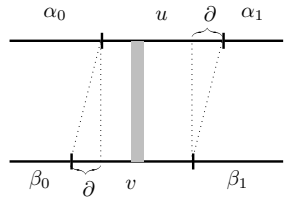
\begin{figure}[htbp]
\centering
\scalebox{.8}{
\begin{tikzpicture}

\draw[line width=0.4mm] [-](0,2) -- (1.5,2)[anchor=south, xshift=-0.75cm, yshift=.2cm]node{$\alpha_0$};
\draw[line width=0.4mm] [|-](1.5,2) -- (3.5,2)[anchor=south, xshift=-1cm, yshift=.2cm]node{$u$};
\draw[line width=0.4mm] [|-](3.5,2) -- (4.5,2)[anchor=south, xshift=-0.5cm, yshift=.2cm]node{$\alpha_1$};

\draw[line width=0.4mm] [-](0,0) -- (1,0)[anchor=south, xshift=-0.5cm, yshift=-.6cm]node{$\beta_0$};
\draw[line width=0.4mm] [|-](1,0) -- (3,0)[anchor=south, xshift=-1cm, yshift=-.6cm]node{$v$};
\draw[line width=0.4mm] [|-](3,0) -- (4.5,0)[anchor=south, xshift=-0.75cm, yshift=-.6cm]node{$\beta_1$};

\fill[lightgray] (2,2) rectangle (2.2,0);

\draw [dotted](1,0) -- (1.5,2);
\draw [dotted](3,0) -- (3.5,2);
\draw [dotted](1.5,0) -- (1.5,2);
\draw [dotted](3,0) -- (3,2);

\draw [decorate, 
    decoration = {calligraphic brace,
        raise=2pt,
 	amplitude= 4pt,mirror
       }] (1,0) --  (1.5,0) node[anchor=north, xshift=-0.2cm,yshift=-0.2cm]{$\partial$};
\draw [decorate, 
    decoration = {calligraphic brace,
        raise=2pt,
 	amplitude= 4pt,
       }] (3,2) --  (3.5,2) node[anchor=south, xshift=-0.2cm,yshift=0.2cm]{$\partial$};

\end{tikzpicture}}
\caption{An edit in the interior of $u$ and $v$.}

\label{fig:jaro}

\end{figure}

 \begin{figure}[htbp]

\centering

\scalebox{.8}{
\begin{tikzpicture}

\draw[line width=0.4mm] [-](0,2) -- (1.5,2)[anchor=south, xshift=-0.75cm, yshift=.2cm]node{$\alpha_0$};
\draw[line width=0.4mm] [|-](1.5,2) -- (3.5,2)[anchor=south, xshift=-1cm, yshift=.2cm]node{$u$};
\draw[line width=0.4mm] [|-](3.5,2) -- (5.5,2)[anchor=south, xshift=-1cm, yshift=.2cm]node{$u$};
\draw[line width=0.4mm] [|-](5.5,2) -- (7.5,2)[anchor=south, xshift=-1cm, yshift=.2cm]node{$u$};
\draw [dotted](7.5,2) -- (9.5,2);
\draw[line width=0.4mm] [|-](9.5,2) -- (10.5,2)[anchor=south, xshift=-0.5cm, yshift=.2cm]node{$\alpha_1$};

\draw[line width=0.4mm] [-](0,0) -- (1,0)[anchor=south, xshift=-0.5cm, yshift=-.6cm]node{$\beta_0$};
\draw[line width=0.4mm] [|-](1,0) -- (3,0)[anchor=south, xshift=-1cm, yshift=-.6cm]node{$v$};
\draw[line width=0.4mm] [|-](3,0) -- (5,0)[anchor=south, xshift=-1cm, yshift=-.6cm]node{$v$};
\draw[line width=0.4mm] [|-](5,0) -- (7,0)[anchor=south, xshift=-1cm, yshift=-.6cm]node{$v$};
\draw [dotted](7,0) -- (9,0);
\draw[line width=0.4mm] [|-](9,0) -- (10.5,0)[anchor=south, xshift=-0.75cm, yshift=-.6cm]node{$\beta_1$};

\fill[lightgray] (2,2) rectangle (2.2,0);
\fill[lightgray] (4,2) rectangle (4.2,0);
\fill[lightgray] (6,2) rectangle (6.2,0);
\fill[lightgray] (8,2) rectangle (8.2,0);

\draw [dotted](1,0) -- (1.5,2);
\draw [dotted](3,0) -- (3.5,2);
\draw [dotted](5,0) -- (5.5,2);
\draw [dotted](7,0) -- (7.5,2);
\draw [dotted](9,0) -- (9.5,2);

\draw [dotted](1.5,0) -- (1.5,2);
\draw [dotted](3.5,0) -- (3.5,2);
\draw [dotted](5.5,0) -- (5.5,2);
\draw [dotted](7.5,0) -- (7.5,2);
\draw [dotted](9.5,0) -- (9.5,2);

\draw [dotted](3,0) -- (3,2);
\draw [dotted](5,0) -- (5,2);
\draw [dotted](7,0) -- (7,2);
\draw [dotted](9,0) -- (9,2);

\draw [decorate, 
    decoration = {calligraphic brace,
        raise=2pt,
 	amplitude= 4pt,mirror
       }] (1,0) --  (1.5,0) node[anchor=north, xshift=-0.2cm,yshift=-0.2cm]{$\partial$};
\draw [decorate, 
    decoration = {calligraphic brace,
        raise=2pt,
 	amplitude= 4pt,
       }] (3,2) --  (3.5,2) node[anchor=south, xshift=-0.2cm,yshift=0.2cm]{$\partial$};

\draw [decorate, 
    decoration = {calligraphic brace,
        raise=2pt,
 	amplitude= 4pt,mirror
       }] (3,0) --  (3.5,0) node[anchor=north, xshift=-0.2cm,yshift=-0.2cm]{$\partial$};
\draw [decorate, 
    decoration = {calligraphic brace,
        raise=2pt,
 	amplitude= 4pt,
       }] (3,2) --  (3.5,2) node[anchor=south, xshift=-0.2cm,yshift=0.2cm]{$\partial$};

\draw [decorate, 
    decoration = {calligraphic brace,
        raise=2pt,
 	amplitude= 4pt,mirror
       }] (5,0) --  (5.5,0) node[anchor=north, xshift=-0.2cm,yshift=-0.2cm]{$\partial$};
\draw [decorate, 
    decoration = {calligraphic brace,
        raise=2pt,
 	amplitude= 4pt,
       }] (5,2) --  (5.5,2) node[anchor=south, xshift=-0.2cm,yshift=0.2cm]{$\partial$};

\draw [decorate, 
    decoration = {calligraphic brace,
        raise=2pt,
 	amplitude= 4pt,mirror
       }] (7,0) --  (7.5,0) node[anchor=north, xshift=-0.2cm,yshift=-0.2cm]{$\partial$};
\draw [decorate, 
    decoration = {calligraphic brace,
        raise=2pt,
 	amplitude= 4pt,
       }] (7,2) --  (7.5,2) node[anchor=south, xshift=-0.2cm,yshift=0.2cm]{$\partial$};

\draw [decorate, 
    decoration = {calligraphic brace,
        raise=2pt,
 	amplitude= 4pt,mirror
       }] (9,0) --  (9.5,0) node[anchor=north, xshift=-0.2cm,yshift=-0.2cm]{$\partial$};
\draw [decorate, 
    decoration = {calligraphic brace,
        raise=2pt,
 	amplitude= 4pt,
       }] (9,2) --  (9.5,2) node[anchor=south, xshift=-0.2cm,yshift=0.2cm]{$\partial$};

\end{tikzpicture}
}
\caption{Words that require arbitrarily large number of edits.}

\label{fig:jaropump}

\end{figure}
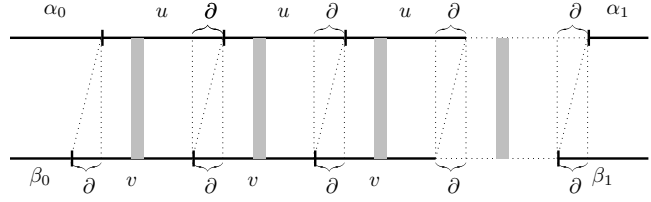

($\rightarrow$): Assume $d_h(\calT_1,\calT_2) = \infty$. Assume $\calA$ has $n$ states and the maximum length of an output produced on any transition or at the end-of-input is $\ell$. Choose a run $\rho$ of $\calA$ such that the distance between the outputs produced on $\rho = \delta_1 \cdots \delta_m$, $m>0$ is at least $((k+2)n+1)\ell$. We can associate each edit in $\lambda_1(\rho)$ with the transition $\delta_i$ such that the edit happens in $\lambda_1(\delta_i)$. Since there are 
$((k+2)n+1)\ell$ edits, there are at least $(k+2)n+1$ transitions in $\rho$ whose output words are edited. 
Associate each transition with its source state. By pigeonhole principle, there is a state $q$ such that $\rho=\rho_1\cdot \rho_2\cdot \rho_3$ where
\begin{enumerate}
\item $\rho_1$ is a run from the initial state $q_0$ to $q$,
\item $\rho_2$ is a run from $q$ to itself,
\item $\rho_3$ is a run from $q$ to a final state $q_f$, and
\item there are at least $(k+1)$ edits in the factor $\lambda_1(\rho_2)$.
\end{enumerate}

Let $u=\lambda_1(\rho_2)$ and $v = \lambda_2(\rho_2)$. Clearly $|u| = |v|$ and $|u| \geq (k+1)$. Since the edits in $u$ are at least $k+1$, there is a position on which the pair $\interior_{\partial_q}(u,v)$ differ.
\end{proof}

\subparagraph*{Deciding \Cref{Claim:hamming}:}
For checking characterisation in \Cref{Claim:hamming}, for each state $q$ for which $M_{q,q}$ is nonempty, we need to check if the $\interior(q) = \{\interior_{\partial_q}(u,v) \mid (u,v) \in M_{q,q}\}$ consist of only identical pairs of words. Towards this, we assume that the transducers given are normalised, i.e., for any transition $\delta$ in $\calA$, $\lambda_i(\delta)$ for $i = \{1,2\}$ is either a single letter or $\epsilon$.

We can verify and compute the unique delay, if exists, for each state in polynomial time. For this, set the delay at intial state to be zero. Now for each transition, compute the delay of the target state. If the delay of the target state is not yet computed, set this as its delay. On the other hand, if delay of target state is already computed, then check the computed delay is equal to that. If not, the delay is not unique and thus, Property ($\star$) is not satisfied and transducers are not close. Otherwise continue with the next transition of $\calA$ till all transitions are covered.

 We construct a forward and backward gadget for each state that unfolds the loop up to its delay in both directions trimming the $\lborder$ and $\rborder$ respectively. For a state $q$ with the delay $\partial_q$, if $\partial_q > 0$, then the forward gadget of $q$ will produce pairs $\{(u[1 \tdots \partial_q], \epsilon) \mid (u,v) \in M_{q,q}, |u| > \partial_q\}$, and the backward gadget of $q$ will produce pairs $\{( \epsilon, v[(n-\partial_q) \tdots n) \mid (u,v) \in M_{q,q}, n=|v| > \partial_q\}$. Symmetrically for $\partial_q <0$.

Instead of constructing these gadgets for each state, we combine the construction for each strongly connected component (SCC). In polynomial time, we can find all the strongly connected components  (SCCs) of the automaton $\calA$. Let $S$ be the set of states in a SCC and $\partial$ be the maximum absolute value of delay among all the states in $S$. Construct a forward and backward gadget (or automaton) $\calF_{S}$ and $\calB_{S}$ respectively with states from $S \times \{-\partial, -\partial +1 , \ldots, 0, 1, 2 , \ldots, \partial\}$ with output functions $\lambda_{f_1},\lambda_{f_2}$ and $\lambda_{b_1},\lambda_{b_2}$ respectively. 

Initially, the forward gadget has states of the form $(q,\partial_q)$ for $q \in S$, and the backward gadget has states of the form $(q,-\partial_q)$for $q \in S$ . 

For a transition $\delta = (p,a,q)$ in $\calA$ where $a\in A$ and $p,q \in S$,
\begin{itemize}
\item In forward gadget, for each state $(p,i) \in \calF_{S}$, we add a transition $\delta_f = ((p,i), a,(q,j))$  such that 
\begin{enumerate}
\item if $i>0$, then $\lambda_{f_1}(\delta_f) = \lambda_1(\delta)$, $\lambda_{f_2}(\delta_f) = \epsilon$, and if  $\lambda_2(\delta) \neq \epsilon$ then $j=i-1$; else $j=i$.
\item if $i<0$, then $\lambda_{f_1}(\delta_f) = \epsilon$, $\lambda_{f_2}(\delta_f) = \lambda_2(\delta)$, and if  $\lambda_1(\delta) \neq \epsilon$ then $j=i+1$; else $j=i$.
\end{enumerate}
\item In backward gadget, for each state $(q,i) \in \calB_{S}$, we add a transition $\delta_b = ((p,j), a,(q,i))$  such that 
\begin{enumerate}
\item if $i>0$, then $\lambda_{b_1}(\delta_b) = \lambda_1(\delta)$, $\lambda_{b_2}(\delta_b) = \epsilon$, and if  $\lambda_2(\delta) \neq \epsilon$ then  $j=i-1$; else $j=i$.
\item if $i<0$, then $\lambda_{b_1}(\delta_b) = \epsilon$, $\lambda_{b_2}(\delta_b) = \lambda_2(\delta)$, and if  $\lambda_1(\delta) \neq \epsilon$ then $j=i+1$; else $j=i$.
\end{enumerate}
\end{itemize}

We construct such a forward and backward gadget for each SCCs.
For each state $q$ in an SCC with a set of states $S$, we construct an automaton $\calA_{S_q}$ by concatenating $\calF_S$, $\calA$ induced with states $S$, and $\calB_S$ with existing transitions and labels. Additionally, following transitions are added.
\begin{itemize}
\item For each state $(p,0) \in \calF_S$, $p \in \calA$, add $((p,0), \epsilon, p)$ with $\epsilon$ output labels.
\item For each state $p \in \calA$, $(p,0) \in \calB_S$, add $(p, \epsilon, (p,0))$ with $\epsilon$ output labels.
\end{itemize}
Let $\lambda_1' = \lambda_{f_1} \cup \lambda_1 \cup \lambda_{b_1}$ and $\lambda_2' = \lambda_{f_2} \cup \lambda_2 \cup \lambda_{b_2}$ be the output labels in $\calA_{S_q}$, and $\omicron_1'$, $\omicron_2'$ maps each state to $\epsilon$ in $\calA_{S_q}$. Let the initial state of $\calA_{S_q}$ being $(q,\partial_q)$ and final state being $(q,-\partial_q)$. It is easy to see that the automaton $\calA_{S_q}$ is unambiguous and is of polynomial size w.r.t.~$\calA$. To verify \Cref{Claim:hamming}, it suffices to check if the functional transducers  $\langle \calA_{S_q}, \lambda_1',\omicron_1'\rangle$ and  $\langle \calA_{S_q}, \lambda_2',\omicron_2'\rangle$ are equivalent for each state $q$ in a SCC with a set of states $S$. 

Next we show closeness w.r.t.~transposition distance. 
We write $u \equiv v$ to denote that words $u$ and $v$ are permutations of each other. The \emph{alphabetic vector} of a word over the alphabet $A$, denoted by $\vec{u}$, is the sequence $(|w|_{a_i})_{a_i\in A}$ for some fixed ordering of $A$. It is easy to observe that two words are permutations of each other if their alphabetic vectors are the same.
\begin{claim}\label{Claim:trans}
Transposition distance between $\calT_1$ and $\calT_2$ is unbounded if and only if one of the following holds
\begin{enumerate}
\item \label{trans1} There is a pair $(u,v) \in M'_{q_0,q_f}$, $q_f\in F$ such that $u\not \equiv v$.
\item \label{trans2} There exists a state $q\in Q$ and $(u,v) \in M_{q,q}$ such that $|u|=|v| > \partial_q$, and $u' \neq v'$ where $(u',v')=\interior_{\partial_q}(u,v)$.
\item \label{trans3} There exists a state $q\in Q$ such that $M_{q,q}$ is infinite, and for each pair $(u,v) \in M_{q,q}$ of length at least $|\partial_q|$, $\interior_\partial(u,v)$ is identical. Further, there are pairs $(u,v)\in M_{q,q}$ and $(\alpha,\beta) \in M_{q_0,q}$ (\textit{resp.}~$M_{q,q_f}$) such that: If $\partial_q\geq 0$, then $\alpha \not \equiv \beta \cdot \lborder(u,v)$ (\textit{resp.}~$ \rborder(u,v) \cdot \alpha  \not \equiv \beta $), and if $\partial_q < 0$, then $\alpha \cdot \lborder(u,v) \not \equiv \beta$ (\textit{resp.}~$\alpha \not \equiv \rborder(u,v) \cdot \beta$).
\end{enumerate} 
\end{claim}

\begin{proof}
($\leftarrow$): It is obvious that if \Cref{trans1} is true, then the transposition distance between $\calT_1$ and $\calT_2$ is unbounded. Therefore we assume that the output pairs of the transducers are permutations of each other. For \Cref{trans2}, the proof is the same as in \Cref{Claim:hamming}.  Next we consider \Cref{trans3}. The cases are symmetric. Assume 
that there exist a pair $(u,v)\in M_{q,q}$, $(\alpha,\beta) \in M_{q_0,q}$, and WLOG $\partial_{q}\geq 0$  such that 
$\alpha \not \equiv \beta \cdot \lborder(u,v)$. Let $(\alpha',\beta')$ be some pair in $M_{q,q_f}$. Consider the pair $(u_i = \alpha u^i\alpha', v_i = \beta v^i\beta')$, $i\geq 1$.  

Let $(x,x) = \interior_{\partial_q}(u,v)$, $z_1 = \lborder(u,v)$, $z_2 = \rborder(u,v)$. 
By assumption $\alpha \not \equiv \beta z_1$, and hence  $z_2\alpha' \not \equiv \beta'$. 
Since interior of $(u,v)$ is $(x,x)$, we can deduce that $\alpha z_2\alpha' \equiv \beta z_1\beta'$. Therefore 
$\vec{\alpha}- \vec{\beta z_1} = \vec{z_2\alpha'}- \vec{\beta'}$. This means that the transpositions have to cancel out the differences in the vectors at each end of the word. We can prove by induction that it requires at least $|x|$ transpositions to mitigate a difference of $1$, while keeping the alphabetic vector of the middle portion the same.
Hence we deduce that $d_t(u_i,v_i) \geq i$.

($\rightarrow$): If $d_t(\calT_1,\calT_2) \in \infty$, either there is a pair of outputs $(u,v)$ such that $d_t(u,v)=\infty$ (This is \Cref{trans1}), or all the output pairs are permutations of each other and there is an infinite set of pairs $S=\{ (u_i,v_i) \mid i>0\}$ such that $d_t(u_i,v_i) \geq i$.  

In the latter case, we show that either \Cref{trans2} or \Cref{trans3} holds. We say the set $S$ is \emph{error-bounded} if there is an $r>0$ such that $u_i$ and $v_i$ differ in at most $r$ positions. 
Clearly, there are sets with bounded errors on which $d_t$ is infinite. We do case analysis.  

If there is an infinite set of pairs $S=\{ (u_i,v_i) \mid i>0\}$ such that $d_t(u_i,v_i) \geq i$ that is \emph{not} error-bounded, we proceed as in the proof of \Cref{Claim:hamming} and obtain \Cref{trans2} by pigeonhole principle.

If the set of all output pairs is error-bounded, then clearly for states $q$ such that $M_{q,q}$ is infinite, the interior of all the sufficiently large pairs in $M_{q,q}$ are identical. Moreover since the output pairs are permutations of each other there is a state $q$ such that $|M_{q,q}| =\infty$ and there is a partial run from $q_0$ to $q$ (or a partial run from $q$ to $q_f$) whose output words are not permutations of each other.  
\end{proof}

\subparagraph*{Deciding \Cref{Claim:trans}:}
The \Cref{trans2} in \Cref{Claim:trans} is same as that of characterisation in  \Cref{Claim:hamming}, and thus can be checked in polynomial time. For checking \Cref{trans3}, it suffices to check if all the paths between any two strongly connected components in $\calA$ have the same alphabetic vectors. Towards this, first topologically sort the strongly connected components in $\calA$.  Consider any two adjacent SCCs with a set of states $S_1$ and $S_2$, respectively in the topological sort of $\calA$ with directed paths from states in $S_1$ to that of $S_2$. Our aim is to check if all the paths from a state in $S_1$ to a state in $S_2$ have the same alphabetic vector. For this, we construct a induced subgraph of $\calA$, say $\calA_I$, with $S_1 \cup S_2$ being the set of states, along with all edges associated with it. Let $S_1'$ be the states in $S_1$ that has an outgoing edge to next SCC in $\calA_I$ and $S_2'$ be the set of states in $S_2$ that has an incoming edge from previous SCC in $\calA_I$. Assume that the induced graph is trimmed, i.e., any state is reachable from a state in $S_1'$ and co-reachable from a state in $S_2'$. For each state in the induced subgraph, we maintain a vector/ sequence of length $|A|$ for some fixed ordering of $A$ where $A$ is the alphabet of $\calA$ such that each position in the vector tells the difference in the count of a letter encountered till now. It can either be a positive value or a negative value. For example, for $A= \{1,b\}$, $(2\ \ \  -3)$ means there are $2$ $a$'s more (or less) in the first (or second) component of the output and there are $3$ $b$'s less (or more) in the first (or second) component of the output seen till now.  Initially, for all states in $S_1'$ in $\calA_I$ we associate a zero vector.  For each transition in $\calA$ involving states in $\calA_I$, we compute the vector of the target state from the computed vector of the initial state. For any state, if there are two different vectors associated, then since $A_I$ is trimmed, there exist a path from $S_1'$ to $S_2'$ with different alphabetic vector. Thus, for each state in $\calA_I$, there is a unique vector. Finally, we need to check if the vector associated with each state in $S_2'$ is a zero vector. Otherwise, there exist a path from $S_1'$ to $S_2'$ with different alphabetic vector. Hence, \Cref{trans3} can be checked in polynomial time. It is easy to see that if both \Cref{trans2} and \Cref{trans3} are not satisfied, then \Cref{trans1} is also not satisfied. 

Therefore, \Cref{Claim:hamming} and \Cref{Claim:trans} can be verified for $\calT_1$ and $\calT_2$ in polynomial time. Thus, closeness of sequential and unambiguous transducers w.r.t.~hamming and transposition distance is decidable in polynomial time.

\subsection{$k$-closeness of Hamming and Transposition distances}
\label{kclosehammingtrans}

We have shown that using \Cref{prop:kapprox}, $k$-closeness is decidable for all metrics in Table~\ref{table:editdistances}. To check if two transducers $\calT_1$ and $\calT_2$ are $k$-close, we check if they have the same domain and they are close w.r.t.~the length metric (otherwise they are neither close nor $k$-close). If so, we check if the domain $T$ is same as the domain of $\kdist{\calT_1}{\calT_2}{d}{k}$.  From the construction used in \Cref{prop:kapprox}, the underlying automaton of $f_{\calT_1,\calT_2}^{\leq k}$ can be of exponential (as it has to keep track of the unprocessed words). So, the complexity of the $k$-closeness algorithm is EXPSPACE. However, for hamming and transposition distance, using the characterisations used for deciding their closeness (see \Cref{Claim:hamming} and \Cref{Claim:trans}), we give a co-NP procedure for deciding $k$-closeness and co-NP $\cap$ NP procedure for computing the distance.
\begin{theorem}
$k$-closeness w.r.t.~Hamming and Transposition distance are decidable for sequential and unambiguous transducers in co-NP time.
\end{theorem}
\begin{proof}
Given two unambiguous transducers, check in polynomial time that if they are close w.r.t.~Hamming (or Transposition) distance. If not, then the distance is $\infty$ and they are not $k$-close. Assume they are close. 

By \Cref{Prop:rattoseq}, it suffices to consider two sequential transducers with a common underlying DFA. Let $\calT_1 = \langle \calA, \lambda_1, \omicron_1\rangle$ and $\calT_2 = \langle \calA, \lambda_2, \omicron_2\rangle$ be two sequential transducers.

Since $T_1$ and $T_2$ are close, from Condition $(2)$ in  \Cref{Claim:hamming} and Condition $2.(b)$ in  \Cref{Claim:trans}, we get that the for any state $q$ in $\calA$, $\interior(q)$ produces only identical pair of outputs. In fact, we can ignore the $\interior$ of the loops, and by only keeping $\lborder$ and $\rborder$ of any loops, we construct a polynomially sized acyclic automaton from $\calA$ such that distance of $\calT_1$ and $\calT_2$ equals to the maximum over the minimum edits required in any accepting path of the acyclic automaton. Towards this, we construct a forward and backward gadget similar to the one constructed for deciding \Cref{Claim:hamming}. However, in this case, instead of trimming the $\lborder$ (\textit{resp.}~$\rborder$) in the forward (\textit{resp.}~backward) gadget, we keep only $\lborder$ (\textit{resp.}~$\rborder$)  and make the other component $\epsilon$. We construct such a forward and backward gadget for each SCCs in $\calA$.  For an SCC, we combine the forward and backward gadgets by merging the states of the form $(p,0)$ in both of these gadgets where $p$ is a state in the SCC. This combined gadget is acyclic.

Given $\calA$, we topologically sort the strongly connected components. We replace each SCC with its combined forward and backward gadget. For any transition $(p,a,q)$ between two adjacent SCC, we add a transition from $((p,-\partial_p), a, (q,\partial_q))$ connecting the combined gadgets of the two SCCs. Therefore, we get a polynomially sized acyclic automaton, say $\calA_T$, from $\calA$ by removing the $\interior$ of any state in $\calA$ and keeping the initial and final states the same as that of $\calA$. Now, it suffices to check if all paths from initial to the final state in the acylic automaton require only at most $k$ edits (substitutions in the case of Hamming distance, while adjacent transpositions in the case of Transposition distance) to convert the output of the first transducer to that of second. The complement of this problem is asking if there exists a path that requires at least $k+1$ edits to convert one output to another. This is in NP since given a certificate , we can polynomially compute the number of edits required to convert one output to another (for hamming - $\mathcal{O}(n)$ and for transposition - $\mathcal{O}(nlogn)$ \cite{cicirello2019kendall} where $n$ is the length of the output words), and hence verify polynomially that whether it requires $k+1$ edits. Since the complement of this problem is in NP, the problem is in co-NP.
\end{proof}

As a consequence, we get the following result.

\begin{proposition}
Computing distance of two transducers w.r.t.~Hamming as well as transposition distance is in co-NP $\cap$ NP.
\end{proposition}

\begin{proof}
First, determinine the closeness in polynomial time. If not close, then distance is $\infty$. If close, then we can construct an polynomially sized acyclic automaton $\calA_T$ such that the distance equals to the maximum over the minimum edits required in any accepting path.  In fact, the distance is bounded by product of the number of transitions in $\calA_T$ and the maximum value in the output length on any single transition. Let $b$ denote the bound on the distance. Now, we can compute the exact distance using $k$-closeness in two ways. One way, is to iterate $k$ over $\{0, 1, \ldots, b\}$, and check for each $k$ that every accepting path in $\calA_T$ requires at most $k$ edits (which is in co-NP). If not, increment $k$ and continue. Else, $k$ is the distance between $T_1$ and $T_2$.  Hence, computing distance is in co-NP. The other way is to iterate $k$ over $\{b, b-1, \ldots , 1\}$, and for each $k$, check if there exists a accepting path with at least $k$ edits (which is in NP). If yes, then $k$ is the required distance between $T_1$ and $T_2$. Otherwise, decrement $k$ and continue. Hence, computing distance is also in NP.
\end{proof}

\section{Discussion and Conclusion}
\label{Sec:conclusion}
 It is shown that  distance between two rational functions w.r.t.~common edit distances is computable. The related notions of diameter of a rational relation, and  the index of a rational relation in the composition closure of another are also computable. We leave open the question of finding the precise computational complexity of the problems in Tables~\ref{table:distance},~\ref{table:diameter} and \ref{table:index}.

The current decision procedure for closeness w.r.t.~conjugacy and Levenshtein family of distances proceeds through the analysis of rational expressions. One could directly work on automata, but it is not enough to check for the conjugacy of simple cycles, as there can be complex strongly connected components. In such cases, a decidability proof for conjugacy can be achieved by utilizing Simon's factorization forests \cite{simon1990factorization} and checking the conjugacy of the factorization trees inductively. Sumfree expressions are doing this in essence, circumventing the need to construct the transition monoids.

Lifting these notions to infinite words, and two-way transducers is an immediate next step. Distance between one-way transducers could be seen as the diameter of a rational relation obtained by the cartesian product. However, when the transducers $\calT,\calS$ are two-way or polyregular, the relation $\{(\calT(w),\calS(w)) \mid w \in \dom(\calT)\}$ need not be rational. It remains to develop techniques for checking the conjugacy of non-rational relations. 

An interesting question is: given two functional transducers $\calT_1$ and $\calT_2$ with bounded distance, does there exist a transducer $\calT$ such that $\calT_2$ is equivalent to a cascading composition of $\calT_1$ and $\calT$? This is often called the \emph{repair problem} and is well-studied between two regular languages \cite{benedikt2011regular}.

\bibliography{reference}

\end{document}